\documentclass[11pt,freqn, reqno,oneside]{amsart}
\usepackage{url}

\usepackage{amsmath,amsthm,amssymb,mathrsfs}

\usepackage[top=1.1in,bottom=1.1in]{geometry} 

\usepackage{color} 
\usepackage{float}

\newcommand{\E}{\mathbb{E}}
\newcommand{\tE}{\tilde{\mathbb{E}}}
\newcommand{\Q}{\mathbb{Q}}

\newcommand{\T}{\mathcal{T}}
\newcommand{\G}{\mathcal{G}}
\newcommand{\OO}{\mathcal{O}}

\newcommand{\PP}{\mathbb{P}}
\newcommand{\tPP}{\tilde{\mathbb{P}}}
\newcommand{\R}{\mathbb{R}}
\newcommand{\Pbs}{P_{BS}}
\newcommand{\Pb}{p}
\newcommand{\Qb}{q}
\newcommand{\Qe}{{\mathbb{Q}^E}}
\newcommand{\Qm}{{\mathbb{Q}^0}}

\newcommand{\tnu}{\tilde{\nu}}
\newcommand{\hnu}{\hat{\nu}}

\newcommand{\cnu}{\check{\nu}}

\newcommand{\cB}{\check{B}}

\newcommand{\ind}{ {1\!\!\mathrm{I}}}

\usepackage{ifpdf}
\ifpdf
\usepackage[pdftex]{graphicx}
\else
\usepackage[dvips]{graphicx}
\fi

\newtheorem{theorem}{Theorem}[section]
\newtheorem{lemma}[theorem]{Lemma}
\theoremstyle{definition}
\newtheorem{definition}[theorem]{Definition}

\newtheorem{proposition}[theorem]{Proposition}
\newtheorem{corollary}[theorem]{Corollary}
\numberwithin{equation}{section}
\theoremstyle{remark}
\newtheorem{remark}[theorem]{Remark}

\begin{document}

\title{European Option Pricing with Liquidity Shocks
}

\author{Michael Ludkovski and Qunying Shen}

\address{Department of Statistics and Applied Probability, University of California,\\
Santa Barbara, CA 93106-3110, USA.}
\email{ludkovski@pstat.ucsb.edu, shen@pstat.ucsb.edu}



\maketitle

\begin{abstract}
We study the valuation and hedging problem of European options in a market subject to liquidity shocks. Working within a Markovian regime-switching setting, we model illiquidity as the inability to trade. To isolate the impact of such liquidity constraints, we focus on the case where the market is completely static in the illiquid regime. We then consider derivative pricing
using either equivalent martingale measures or exponential indifference mechanisms.  Our main results concern the analysis of the semi-linear coupled HJB equation satisfied by the indifference price, as well as its asymptotics when the probability of a liquidity shock is small. We then present several numerical studies of the liquidity risk premia obtained in our models leading to practical guidelines on how to adjust for liquidity risk in option valuation and hedging.
\end{abstract}

\keywords{Keywords: liquidity shock; indifference price; exponential utility maximization}

\section{Introduction}
The traditional option pricing theory relies on a key assumption: that markets are always liquid and agents can trade whenever they wish. However, this is not generally true in real financial markets \cite{longstaff2004financial}. There is increasingly a common phenomenon of \emph{liquidity crises} in even well-established securities markets, where liquidity dries out or trading is virtually impossible for a number of reasons, such as political turmoil, war, and financial crises. Two prominent examples include the terrorist attacks of 9/11 which closed all US security exchanges for four days, and the ``flash crash'' of May 6, 2010 \cite{Kirilenko11} when numerous stocks experienced extreme price swings (up to 99\% of value) coupled with minimal trading volume during the course of several hours.

In recent years a number of approaches for dealing with market illiquidity within \emph{optimal investment} frameworks have been developed. Schwartz and Tebaldi \cite{schwartz2006illiquid} considered the optimal portfolio problem with permanent trading interruptions; Diesinger et al. \cite{diesinger2008asset} addressed terminal wealth maximization for CRRA utilities. Ludkovski and Min \cite{ludkovski2010illiquidity} and Gassiat et al.~\cite{PhamGassiat11} analyzed the impact of illiquidity on optimal consumption strategies. 
In this paper we extend this analysis to the problem of \emph{pricing} of derivative securities. To our knowledge, there is no previous analysis for effect of illiquidity on option prices.

In this paper, we define illiquidity as the inability to trade in a timely way. While liquidity crises involve a host of phenomena that affect the market environment, including downward jumps in asset prices, modified asset dynamics and restrictions on trading strategies, here we concentrate on the last item. Namely, to isolate the impact of liquidity shocks,  we consider the extreme scenario whereby (i) trading in assets is completely suspended during the shock and (ii) the underlying asset price remains frozen while the shock is ongoing.

We adopt a regime-switching description of market liquidity which is modeled through a separate Markovian liquidity factor. This aspect of our model connects with the growing literature on Markov-modulated financial markets driven by the recognition that the economic environment is not stable. Some of the pertinent analyses include  \cite{blanchet2005dynamic,elliott2005option, houssou2010indifference, leung2010markov,siu2009option}. The model of Blanchet et al.~\cite{blanchet2005dynamic} effectively corresponded to the case where the liquidity shock was permanent, while Siu  et al.~\cite{siu2009option} and Leung~\cite{leung2010markov} studied  valuation of options using Esscher and minimal entropy martingale measures, respectively.

The main rationale for our choices is to have a well-understood benchmark for comparing derivative valuations and to focus on the effect of trading suspension. In particular, thanks to the special form of liquidity shocks, we can  directly compare to the classical frictionless Black-Scholes model, specifically linking to the options' time-decay. Indeed, a liquidity shock induces an instantaneous (in terms of business-time) jump in the option's Delta, making perfect hedging impossible and exposing investors to \emph{charm} risk, namely the dependence of Delta on time-to-maturity. In reality, asset valuations continue to change (often in an unfavorable way to the investor) and some form of trading may still be possible (e.g., only a short-selling restriction). However, the resulting model is simply somewhere between the classical textbook  description and the one considered herein; similarly, incorporating other features, such as jumps in asset prices, only obscures the final conclusions.

Potential liquidity shocks introduce a new non-traded source of risk and make the asset market incomplete.
 Derivative pricing in incomplete markets largely follows two main approaches. First, one may fix an equivalent martingale measure (EMM) $\Q$ and calculate option prices as expectations under $\Q$.
Some of the common candidates for EMMs include minimal martingale measure \cite{follmer1990hedging}, minimal entropy martingale measure \cite{frittelli2000minimal,miyahara1995canonical}, variance-optimal martingale measure \cite{schweizer1996approximation}, and empirical martingale measure \cite{blanchet2005dynamic}. Second, one may use a utility maximization criterion to obtain indifference prices that take into account the risk aversion of the investor \cite{becherer2004utility,henderson2009utility,houssou2010indifference,ilhan2004portfolio, leung2010markov,rouge2000pricing,zhou2006indifference}. Given the widely reported crash-o-phobia of liquidity crises, this is a useful nonlinear pricing rule complementing the classical expectations-based approach. In this paper we compare both approaches for determining the option's liquidity premia, with a special focus on exponential utility valuation, which leads to wealth-independent indifference prices and natural connections to the minimal entropy measure.

Our main contributions concern the analysis of the Hamilton-Jacobi-Bellman (HJB) equations satisfied by the indifference prices, including explicit asymptotics in terms of the two main model parameters: level of risk-aversion and probability of liquidity shock.
We also present a detailed comparative analysis between the model prices and the classical Black-Scholes prices with the aim of providing rules of thumb for incorporating liquidity shock risk into the valuations. For this purpose, we introduce the new concepts of \emph{implied} and \emph{adjusted} time to maturity that allow translation of liquidity risk premia into the Black-Scholes language.

The rest of paper is organized as follows: in Section \ref{sec:model} we setup the probabilistic model. In Section \ref{sec:emm-pricing} we summarize the EMM approach to pricing under liquidity shocks; Section \ref{sec:utility-max} then describes the utility indifference approach to derivative valuation. Section \ref{sec:asympt} discusses asymptotic analysis of the utility-based pricing mechanism to gain further insight into our formulas; in Section \ref{sec:implied-ttm} we present a handy interpretation of the liquidity premium through implied time-to-maturity. All these results are illustrated in
Section \ref{sec: numeric} with numerical experiments. Most proofs are delegated to the Appendix.

\section{The Market Model}\label{sec:model}
In this section, we describe the market model in the presence of potential liquidity shocks. Fix a complete probability space ($\Omega, \mathcal{F},\PP$), where $\PP$  is the real-world probability. We consider a finite investment time horizon $T < \infty$, which is chosen to coincide with the expiration date of all securities in our model, and a financial market which consists of a stock and a risk-free asset (cash), and is subject to the liquidity shocks. We use $(X_t, \pi_t, S_t)$ defined on $\R\times\R\times\R^+$ to denote the total wealth, stock holdings as proportion of total wealth, and the stock price at time $t \in [0,T]$, respectively.
%
\subsection{Market Dynamics}
We assume the market has two states, liquid (0) and illiquid (1).  We use a continuous-time Markov chain $(M_t)$ to represent the changing liquidity of the financial market with state space $E=\{0,1\}$. The Markov chain $(M_t)$ is a proxy for the fluctuating market liquidity and modulates asset dynamics. Its infinitesimal generator is
\begin{align} \label{eq: A}
A = \begin{pmatrix}
-\nu_{01} & \nu_{ 01}\\
\nu_{10}  & -\nu_{10} \end{pmatrix},
\end{align}
where $\nu_{01}$ and $\nu_{10}$ are for simplicity constants.  Associated with $(M_t)$ are two counting processes $(N_{01}(t))$ and $(N_{10}(t))$ with intensity rates of $\nu_{01} \ind_{\{M_t = 0\}}$ and $\nu_{10} \ind_{\{ M_t = 1\}}$, respectively counting the transitions $0 \rightarrow 1$ and $1 \rightarrow 0$.

We assume that the market value of the stock follows a Markov-modulated Geometric Brownian motion (GBM) model. Without loss of generality, we assume a zero interest rate which is equivalent to working with discounted price processes. In the \textbf{liquid state ($M_t = 0$)} the market dynamics follow the classical Black-Scholes model, and the stock can be traded continuously and frictionlessly. More precisely, we take
\begin{subequations}\label{eq:liquid-dynamics}
\begin{align}
dS_t&=\mu_0 S_tdt+\sigma_0 S_tdW_t, \\
dX_t&=  \mu_0 \pi_t X_t dt+\sigma_0 \pi_t X_t dW_t,
\end{align}
\end{subequations}
where $\mu_0$ and $\sigma_0$ denote the stock drift and volatility and  $(W_t)$ is a standard one-dimensional Brownian motion under $\PP$, which is assumed to be independent of the Markov chain $(M_t)$.
In the \textbf{illiquid state ($M_t = 1$)} the market is static and trading in stock is not permitted: $$dS_t = dX_t = 0.$$

As can be seen, this market model has two sources of uncertainty, one generated by  $(W_t)$ and the other by $(M_t)$.
The liquidity constraint is captured by our set $\mathcal{A}$ of admissible trading strategies. Namely, letting $\mathcal{G}_t := \sigma( S_s, M_s : s \le t)$, a strategy $(\pi_t) \in \mathcal{A}$ is admissible if it is self-financing, $\mathcal{G}$-progressively measurable, satisfies the integrability condition $\E[ \int^T_0\pi^2_t\,dt] < \infty$, and the constraint $d\pi_t = 0$ on $\{t: M_t = 1\}$.

We focus on pricing European contingent claims with maturity date $T$ and payoff $h(S_T)$. For simplicity, we assume that the terminal payoff does not depend on $M_T$; it would be straightforward to incorporate liquidity penalties for option exercise. For later use we recall the complete Black-Scholes market counterpart which has the unique martingale measure $\tPP$ with
\begin{align}\label{eq:gbm}
dS_t = \sigma_0 S_t d\tilde{W}_t,
\end{align}
and where contingent claim prices are given by the classical no-arbitrage formula (note that we use the \emph{time-to-maturity} parametrization)
\begin{align}
\Pbs(T, S) \triangleq \tE \left[ h(S_T) \big| \, S_0 = S \right],
\end{align}
where $\tE \equiv \E^{\tPP}$.

\begin{remark}\label{rem:adding-L}
More generally (see \cite{diesinger2008asset,ludkovski2010illiquidity}) the market need not be static during the illiquid periods.
Thus, one may postulate modified dynamics for the stock during the liquidity shock and also include instantaneous price drops when the liquidity regime changes. For instance, assuming that when $M_t = 1$  the stock follows a GBM with drift $\mu_1$ and volatility $\sigma_1$ and that it experiences a fixed drop of $L\%$ at the beginning of each liquidity shock,  the  analogue of \eqref{eq:liquid-dynamics} would be $S_t = (1-L)S_{t-}$, $X_t = (1-L\pi_{t-})X_{t-}$, and $\pi_t=\pi_{t-}(1-L)/(1-\pi_{t-}L)$ whenever $M_t - M_{t-}=1$ and
\begin{subequations}\label{eq:illiquid-dynamics}
\begin{align}
dS_t&=\mu_1 S_t \, dt+\sigma_1 S_t \, dW_t, \\
dX_t&=  \mu_1 \pi_t X_t\, dt+\sigma_1 \pi_tX_t\, dW_t,\\
d\pi_t &= \left\{\mu_1 \pi_t (1-\pi_t)+\sigma_1^2\pi_t^2(2\pi_t-1)\right\} \, dt + \sigma_1 \pi_t (1-\pi_t) \, dW_t.
\end{align}
\end{subequations}
However, the resulting model is significantly more complex due to the fact that $\pi_t$ is changing during the illiquid regime and hence remains part of the system state. In other words, the state variables when $M_t =1$ are $(\pi_t, X_t, S_t)$ introducing an extra state dimension. We further discuss the resulting pricing equations in Remark \ref{rem:pricing-L}.

\end{remark}

\subsubsection*{Notations}

We use $\E^\PP_{t,S,i}, i=0,1,$ to denote expectation taken under the measure $\PP$ with starting values $S_t =S, M_t=i$. We also let
\begin{align}
\beta(i) & \triangleq \left\{ \begin{aligned} \frac{\mu_0}{\sigma_0} & \quad\text{ if } i=0,\\
0 & \quad\text{ if } i=1,
\end{aligned} \right.\label{eq: sharp} \\
D&=\begin{pmatrix}
\frac{\mu_0^2}{2\sigma_0^2} & 0\\
0  & 0\end{pmatrix},
\end{align}
so that $\beta$'s are the Sharpe ratios in each state (note: the zero risk-free rate).

\section{Option Pricing via Martingale Measures}\label{sec:emm-pricing}
In this section, we summarize the option valuation problem in the context of the EMM methodology. We mainly follow Leung~\cite{leung2010markov} who considered a related regime-switching market.
We begin by characterizing the set $\mathcal{E}$ of all equivalent local martingale measures (EMMs) in this liquidity switching market. A probability measure $\Q$ in $\mathcal{E}$, equivalent to ${\PP}$, should be regarded as a risk neutral measure with respect to both asset price and the liquidity risk. Using Girsanov's Theorem for the Wiener process $(W_t)$, and the Markov Chain $(M_t)$  
we obtain (cf.~Theorem 7 in Leung \cite{leung2010markov}):
\begin{proposition} \label{allEMM}
The set $\mathcal{E}$ of all possible EMMs is given by
\begin{align}\notag
\mathcal{E}& = \Bigl \{\Q^{\alpha}: \exists \{\alpha_t(i,j)\}_{i \neq j, i,j\in E}\ge 0 \textrm{ bounded and $\G$-adapted s.t. } \\ & \quad \frac{d{\Q}^{\alpha}}{d\PP}\Big\vert_{\mathcal{G}_t}=\xi_0(t) \times \xi^\alpha_1(t); \,\E^{\PP}[\xi_0(t) \xi^\alpha_1(t)]=1, \forall t \in [0,T]\Bigr  \},   \notag \qquad\textrm{where } \\
&\xi_0(t) \triangleq \exp \left (-\frac{1} {2} \int^t_0 \beta^2(M_s)ds-\int^t_0\beta(M_s)dW_s\right ),  \notag \\  \label{eq: allEMM}
&\xi^\alpha_1(t) \triangleq \exp \left (-\int^T_0  (\tilde{A}_s(M_s, M_s)- A_s(M_s, M_s))ds \right) \cdot \!\!\prod_{\substack{0 \le s \le t: \\ M_{s-}\neq M_s}} \alpha_s(M_{s-},M_s), \\
&\text{with } \quad \tilde{A}_t(i,j)=\left \{ \begin{array}{rcl} \alpha_t(i,j)A(i,j) & \textrm{ if }& i \neq j, \nonumber \\
                                       -\sum_{k \neq i}\tilde{A}_t(i,k)   & \textrm{ if }& i = j.  \end{array} \right. \qquad\quad \text{and A given in \eqref{eq: A}.}
\end{align}
\end{proposition}

In Proposition \ref{allEMM}, the exponential martingale $(\xi_0(t))$ is the Girsanov transform that makes $(S_t)$ a $\Q^{\alpha}$-martingale, while $(\xi^\alpha_1(t))$ is the exponential martingale that transforms the generator of $(M_t)$.  To preserve the Markovian property of $(M_t)$ under the new measure ${\Q}^{\alpha}$, we henceforth require that all $(\alpha_t(i,j))$ be Markovian. Then,  under $\Q^{\alpha}$, the generator matrix of $(M_t)$ is $\tilde{A}=[\tilde{A}_t(i,j)]_{i,j\in E}$. The collection $\{\alpha_t(i,j)\}_{i \neq j}$ can be considered as the risk premium factors for the liquidity switching risk \cite{leung2010markov}.

Using the law of iterated expectations by expressing the price as a discounted risk neutral expectation over all possible paths of the liquidity state $(M_t)$ we obtain
\begin{proposition} \label{prop_EMMprice}
Given a measure $\Q^{{\alpha}} \in \mathcal{E}$, the option price $p^{{\alpha}}(t,S,i)$ is given by
\begin{align} \label{eq: gEMMprice}
p^{\alpha}(t,S,i) &= \E^{{\alpha}}_{t,S,i} \left[ h(S_T) \right] = e^{-\int^T_t\tilde{\nu}_{i,1-i}(u)du} P_{BS}(T-t,S) \nonumber \\ & \quad + \int^T_t\tilde{\nu}_{i,1-i}(\tau) e^{-\int^\tau_t \tilde{\nu}_{i,1-i}(u)du}\E^{\alpha}_{t,S,i}[ p^\alpha(\tau, S_\tau, 1-i)]d\tau,
\end{align}
where $\tilde{\nu}_{i,j}(t) \equiv \tilde{A}_t(i,j) = \alpha_t(i,j)A(i,j)$, $i \neq j$.
\end{proposition}
A formal proof of Proposition \ref{prop_EMMprice} is given in \ref{app_EMMprice}.


In the special case where $\alpha \equiv \mathbf{1}$ is a unity transformation, we obtain the Minimal Martingale Measure (MMM)
$\Qm$ (with $\frac{d\Qm}{d\PP}\big|_{\mathcal{G}_t}= \xi_0(t)$) \cite{elliott2005option,blanchet2005dynamic}, which corresponds to a zero risk premium associated with the liquidity risk. By Girsanov's theorem, under $\Qm$, $(S_t)$ follows
\begin{align}
dS_t &=1_{\{M_t =0\}} \sigma_0 S_t\,d\tilde{W}_t,
\end{align}
where $(\tilde{W}_t)$ is a standard Brownian motion under $\Qm$ and the $\Qm$-generator of $(M_t)$ remains $A$. We use $p^{MM}(t,S)$ to denote the price under MMM conditional on $M_t = 0$.
Another special case is furnished by the minimal entropy martingale measure (MEMM), $\Qe$.

\begin{definition} The minimal entropy martingale measure, $\Qe$, minimizes the relative entropy with respect to $\PP$ over the set of EMMs $\mathcal{E}$:
\begin{align}
&\Qe={\arg\min}_{{\Q} \in \mathcal{E}} H_{T}({\Q} \vert \PP), \qquad\text{where}\\
&H_{T}(\Q\vert \PP) \triangleq \left \{ \begin{array}{rl} \E^{\Q}\left [ \log\frac{d{\Q}}{d\PP} | \,\mathcal{G}_T \right], & {\Q} \ll \PP,  \\
                                      +\infty, & \textrm{ otherwise }.  \end{array} \right. \label{eq:H}
\end{align}
\end{definition}

\begin{definition}\label{de_Fs} Consider the ODE system $F'(t)=(D-A)F(t)$, with $F(t) \equiv (F_1(t),F_2(t))$ and  $F(T)=[1, 1]^\top$. The solutions $F_0(t), F_1(t)$ are explicitly given by
\begin{align}
F_0(t) &= c_1e^{\lambda_1 t}+c_2e^{\lambda_2 t}, \\
F_1(t) &= \frac{1}{\nu_{01}}\left\{ c_1(d_0+\nu_{01}-\lambda_1)e^{\lambda_1 t}+c_2(d_0+\nu_{01}-\lambda_2)e^{\lambda_2 t} \right\}, \\
\qquad\text{where}\quad d_0 & \triangleq D(1,1)=\frac{\mu_0^2}{2\sigma_0^2}, \\
\lambda_{1,2} &=\frac{d_0+\nu_{01}+\nu_{10}\pm \sqrt{(d_0+\nu_{01}+\nu_{10})^2-4d_0\nu_{10}}}{2}, \nonumber \\
c_1 &= \frac{\lambda_2-d_0}{\lambda_2-\lambda_1}e^{-\lambda_1T}, \qquad\text{and}\qquad c_2 = \frac{\lambda_1-d_0}{\lambda_1-\lambda_2}e^{-\lambda_2T}. \nonumber
\end{align}
\begin{remark} \label{re_Fs}
Let $F_2(t) =  e^{-d_0(T-t)}$. Through simple algebra, we have $0 < \lambda_2 < d_0 < \lambda_1$, which leads to $F_1(t) > F_0(t) > F_2(t) $, $\forall t \in [0,T]$.
\end{remark}
\end{definition}

\noindent Applying Theorem 9 in Leung~\cite{leung2010markov} we immediately obtain the following characterization of $\Qe$.
\begin{proposition} \label{the_MEMM}  The MEMM is $\Qe={\Q}^{\alpha^E}$, where the minimal entropy risk factor ${\alpha}^E$ is given by ${\alpha}^E(i,j)=\frac{F_j(t)}{F_i(t)}, i,j \in \{0,1\}$. Thus,
 under $\Qe$ the Markov chain $(M_t)$ becomes a time-inhomogeneous Markov chain with transition intensities
%
\begin{align}\label{eq:Ae}
\hnu_{i,1-i}(t) = \nu_{i,1-i} \frac{F_{1-i}(t)}{F_i(t)}, \qquad i=0,1.
\end{align}
\end{proposition}

Proposition \ref{the_MEMM} presents an explicit form of the MEMM. By Remark \ref{re_Fs}, under $\Qe$, $(M_t)$ is more likely to be in state 1, $\hnu_{01} > \nu_{01}, \hnu_{10} < \nu_{10}$. That is, under the MEMM $\Qe$, the market is more \emph{vulnerable} to a lengthier illiquid shock.

\begin{corollary}\label{MEMMprice}
By Proposition \ref{prop_EMMprice},  the price under the minimal entropy martingale measure $\Qe$, denoted $p^E(t,S)$ for the case $M_t = 0$, is given by substituting $\tnu_{01}$ with $\hnu_{01}$, and $\tnu_{10}$ with $\hnu_{10}$ in formula (\ref{eq: gEMMprice}).
\end{corollary}

\section{Exponential Indifference Pricing}\label{sec:utility-max}
To account for investor risk preferences, we next investigate a nonlinear pricing mechanism in the presence of liquidity shocks.
In this section, we derive the exponential indifference price of $h(S_T)$ using methods of stochastic control and nonlinear partial differential equations. Throughout this analysis, the exponential utility function is used:
\begin{align}
u(x)=-e^{-\gamma x},
\end{align}
where $\gamma >0$ is the investor's risk aversion parameter. Exponential utility has various advantages, such as yielding prices that are wealth independent, and being related via the dual representation to the minimal entropy martingale measure.

A starting point for the approach of utility-based pricing are the value functions $\hat{U}^i(t,X,S)$  for the optimal investment problem in the presence of a contingent claim,
\begin{align}\label{defn:Uhat}
 \hat{U}^i(t,X,S) \triangleq \sup_{(\pi_t) \in \mathcal{A}} \E^\PP_{t,X,S,i} \left[-e^{-\gamma (X_T+h(S_T))} \right], \qquad i=0,1.
 \end{align}
Standard stochastic control methods imply the following:

\begin{proposition}\label{prop_hatU}
The optimal value functions $\hat{U}^i(t,X,S)$ with $(t,X,S) \in [0, T] \times \R \times\R^+$  for a holder of a bounded contingent claim paying off $h(S_T)$ at terminal time $T$ are given by
$$
\hat{U}^i(t,X,S) = -e^{-\gamma X} e^{-\gamma R^i(t,S)},
$$
where $R^i(t,S)$ are the unique viscosity solutions of the coupled semi-linear system
\begin{align} \label{eq: Rhat}
\left\{ \begin{aligned}
R^0_t  +  \frac{1}{2}\sigma_0^2 S^2R^0_{SS} - \frac{\nu_{01}}{\gamma}e^{-\gamma (R^1-R^0)} + \frac{(d_0+\nu_{01})}{\gamma}=0, \\
R^1_t -\frac{\nu_{10}}{\gamma}e^{-\gamma (R^0-R^1)}+\frac{\nu_{10}}{\gamma}=0,
\end{aligned}\right.
\end{align}
with the terminal conditions $\hat{R}^i(T,S)=h(S)$, $i = 0, 1$.
\end{proposition}

As a Corollary, taking $h(S) \equiv 0$, we recover the value functions $\hat{V}^i(t,X)$, $i=0,1,$ for the Merton problem of optimal investment in the liquidity switching market.

\begin{corollary}\label{prop_newMerton}
The value function for optimal investment subject to illiquidity shocks is
\begin{align}\label{eq: Vhat}
\hat{V}^i(t,X)& \triangleq \sup_{(\pi_t) \in \mathcal{A}} \mathbb \E^\PP_{t,X,i}[-e^{-\gamma X_T}]=-e^{-\gamma X}F_i(t), \qquad i=0,1, 
\end{align}
where $F_0(t)$ and $F_1(t)$ are given in  (\ref{de_Fs}).
\end{corollary}

\begin{proof} This is a special case of  Proposition \ref{prop_hatU} with zero options. We observe that if the terminal condition is zero, the solutions of \eqref{eq: Rhat} are independent of $(S_t)$, $R^i = R^i(t)$. Substituting $F_i(t) = e^{-\gamma R^i(t)}$ and simplifying then yields \eqref{de_Fs} which are easily seen to be smooth.  The link between the Merton problem and relative entropy minimization that was originally used to derive $F_i$'s in \eqref{de_Fs} is well-known, see e.g.~\cite[Theorem 11]{leung2010markov}.
\end{proof}

\begin{remark} Recall that without liquidity shocks, i.e., $\nu_{01}=0$, the Merton value function is
$
\hat{V}(t,X)= -e^{-\gamma X}F_2(t).
$
The possibility of illiquidity  increases the risk associated with the investment behavior. As one would expect,  the discount rate has to be therefore increased relative to the perfectly liquid case,  $F_i > F_2$, see  Remark \ref{re_Fs}.
\end{remark}

\subsection{Iterative Approximation}
Since we are unable to provide an \emph{a priori} regularity estimate on $R^i$ in \eqref{eq: Rhat}, we instead pursue a decoupling approach. The main idea is to approximate $R^{0}$ and $R^1$ through a fixed-point iterative scheme yielding \emph{smooth} approximations $R^{(i,k)}$, $i=0,1$, $k=1,2,\ldots$. We show that $R^{(i,k)}$ converge to $R^i$ uniformly on compacts, and provide a characterization of $R^{(i,k)}$ as a unique classical solution of the corresponding (de-coupled) HJB equations. The following Theorem summarizes our results, with the proof in \ref{app_HJB}.

\begin{theorem}\label{thm:classical-soln}
There exist functions $\hat{U}^{(i,k)}(t,X,S)$, $i=0,1$, $k=0,1,\ldots$ such that:
\begin{enumerate}
\item $\hat{U}^{(i,k)}(t,X,S) \in C^{1,2,2}([0,T] \times \R \times \R_+)$ are classical solutions of the HJB equations \eqref{eq: JPDE};

\item $\hat{U}^{(i,k)}(t,X,S)$ admit the control representation
\begin{align}\label{def:hatU_ik}
\hat{U}^{(i,k)}(t,X,S)=\sup_{(\pi_t) \in \mathcal{A}}\E^\PP_{t,X,S,i} \left[u(X^{(i,k)}_T+h(S^{(i,k)}_T)) \right], \quad i=0,1;
\end{align}

\item For any compact subsets $K_x \subset \mathbb{R}$ and $K_S \subset (0, \infty)$ we have
\begin{align}\label{eq:converge}
\lim_{k\rightarrow \infty}\sup_{(t,X,S) \in[0, T] \times K_x \times K_S}\vert  \hat{U}^{(i,k)}(t,X,S)-\hat{U}^i(t,X,S)\vert =0.
\end{align}
\end{enumerate}
\end{theorem}

\subsection{Indifference Prices}
Using $\hat{U}^i$ and $\hat{V}^i$,
the buyer's indifference prices $\Pb$ (initial state 0) and $\Qb$ (initial state 1) are defined via
\begin{align} \label{eq: bpricedef}
\hat U^0(t,X-\Pb,S) &=\hat V^0(t,X),\\
 \notag \hat U^1(t,X-\Qb,S) &=\hat V^1(t,X).
\end{align}
As is well known, under exponential utility the indifference price is wealth independent, $\Pb = \Pb(t,S)$, and $\Qb=\Qb(t,S)$. We know the value functions $\hat{V}^i=-e^{-\gamma X}F_i(t)$ from (\ref{eq: Vhat}) and $\hat{U}^i(t,X,S) =-e^{\gamma X}e^{-\gamma R^i(t,S)}$ from Theorem \ref{thm:classical-soln}, so combining these with the indifference price definition and simplifying equation (\ref{eq: Rhat}), we obtain the following PDEs for $\Pb$ and $\Qb$:

\begin{proposition} \label{prop_bprice}The option buyer's prices are the unique viscosity solutions of the semi-linear elliptic PDEs
\begin{align} \label{eq: bpricepde}
\left\{ \begin{aligned}
\Pb_t & + \frac{1}{2}\sigma_0^2 S^2\Pb_{SS}-\frac{\nu_{01}}{\gamma}\frac{F_1}{F_0}e^{-\gamma (\Qb-\Pb)}+\frac{(d_0+\nu_{01})}{\gamma}-\frac{1}{\gamma}\frac{F_0'}{F_0}=0,\\
\Qb_t & -\frac{\nu_{10}}{\gamma}\frac{F_0}{F_1}e^{-\gamma (\Pb-\Qb)}+\frac{\nu_{10}}{\gamma}-\frac{1}{\gamma}\frac{F_1'}{F_1}=0,
\end{aligned}\right.
\end{align}
with terminal conditions $\Pb(T,S)=\Qb(T,S) =h(S)$.
\end{proposition}
\noindent We note that all the terms involving $F_i(t)$  in \eqref{eq: bpricepde} are bounded on $[0,T]$.

\begin{remark} One can similarly consider the writer's indifference price $p^w(t,S)$ and $q^w(t,S)$, 
which then solve
\begin{align} \label{eq: wpricepde}
\left\{ \begin{aligned}
p^w_t & + \frac{1}{2}\sigma_0^2 S^2p^w_{SS}+\frac{\nu_{01}}{\gamma}\frac{F_1}{F_0}e^{\gamma (q^w-p^w)}-\frac{(d_0+\nu_{01})}{\gamma}+\frac{1}{\gamma}\frac{F_0'}{F_0}=0, \\
q^w_t &+\frac{\nu_{10}}{\gamma}\frac{F_0}{F_1}e^{\gamma (p^w-q^w)} - \frac{\nu_{10}}{\gamma}+\frac{1}{\gamma}\frac{F_1'}{F_1}=0,
\end{aligned}\right.
\end{align}
with terminal condition: $p^w(T,S) =q^w(T,S)=h(S)$.
\end{remark}

The utility indifference pricing procedure also provides an explicit identification of the hedging position, which is found naturally as part of the optimization problem, namely through a maximizer $\pi^*(t,S)$ of \eqref{defn:Uhat}. In an incomplete market perfect hedging is not possible, but this utility indifference hedging strategy provides a dynamic optimal adjustment of the portfolio strategy for an investor who holds the contingent claim.  For exponential utility, the optimal trading strategy $\pi^*(t,S)$ satisfies
\begin{align}\label{eq:pi-star}
\pi^*(t,S) X= \frac{\mu_0}{\sigma_0^2\gamma}-S \frac{\partial \Pb(t,S)}{\partial S},
\end{align}
in which the first term is exactly from the classical Merton problem, and the second term $\frac{\partial \Pb(t,S)}{\partial S}$ is the counterpart of the classical Delta to account for hedging motives of $h(S_T)$ by the risk-averse investor. We discuss \eqref{eq:pi-star} in Section \ref{sec:hedge}. 

\begin{remark}[Duality Representation]
Recall relative entropy defined in \eqref{eq:H}. One obtains (see  equation (4.6) in \cite{delbaen2002exponential}) the following indifference price expression:
\begin{align}\label{eq:dual-price}
\Pb(t,S;\gamma) &= \inf_{\Q \in \mathcal{E}}\left\{\E^\Q\left[h(S_T)\right]+\frac{1}{\gamma}\left (H(\Q\vert \PP)-H(\Qe \vert \PP)\right )\right \}.
\end{align}
The representation \eqref{eq:dual-price} shows that the indifference price is decreasing in $\gamma$, and yields the risk aversion asymptotics  (see Proposition (1.3.4) in \cite{becherer2001rational})  :
\begin{align}\label{eq:gamma-zero-limit}
\lim_{\gamma \rightarrow \infty}\Pb(t,S; \gamma)&=\inf_{\Q \in \mathcal{E}}\E^{\Q} \left[h(S_T)\right],\\
\lim_{\gamma \rightarrow 0}\Pb(t,S; \gamma)&= \E^{{\Qe}}\left[h(S_T) \right].
\end{align}
This implies  that when the risk aversion vanishes, the indifference price converges to the MEMM price. When the risk aversion goes to infinity (total risk aversion), we obtain the super-replication price. 
\end{remark}

\begin{remark}\label{rem:pricing-L}
Continuing Remark \ref{rem:adding-L} one could consider indifference pricing of $h(S_T)$ assuming more general dynamics in the illiquid regime. However as mentioned before, this would require adding $\pi_t$ as a state variable when $M_t =1 $ (i.e.~for the $\Qb$-value function in \eqref{eq: bpricepde}). We present details of the corresponding generalizations of \eqref{eq: bpricepde} in \ref{app:L}.
\end{remark}

\section{Asymptotic Analysis}\label{sec:asympt}
The indifference prices from \eqref{eq: bpricepde} do not  have a closed-form analytical solution. To gain further insight,
in this section we perform regular perturbations of \eqref{eq: bpricepde} in terms of a small risk aversion parameter $\gamma$ or a small probability of liquidity shock $\nu_{01}$.

\subsection{Small Risk Aversion Asymptotics}
For the asymptotic case $\gamma \rightarrow 0$, we can reduce the non-linear reaction-diffusion indifference price equation (\ref{eq: bpricepde}) to a linear one,  obtaining a closed-form representation for the limiting price.

We apply the formal asymptotic expansion
\begin{align*}
p(t,S)=p^{(0)}(t,S)+\gamma p^{(1)}(t,S) + \ldots \quad\text{and}\quad
\Qb(t,S)=\Qb^{(0)}(t,S)+\gamma \Qb^{(1)}(t,S) +\ldots.
\end{align*}

\begin{theorem}\label{p0p1}
Let $B_{ij}(s,t) \equiv \exp(- \int_s^t \hnu_{ij}(u) \,du)$ for $i=0,1$ and $j=1-i$.
The expansion terms $p^{(0)}$, $p^{(1)}$, $q^{(0)}$, and $q^{(1)}$ admit the stochastic representations
\begin{align} \label{eq: bp0}
&\left\{ \begin{aligned}
\Pb^{(0)}(t,S)&= B_{01}(t,T) P_{BS}(T-t,S) + \int^T_t \!\hnu_{01}(\tau) B_{01}(t,\tau) \tE_{t,S}\!\left[\Qb^{(0)}(\tau,S_\tau) \right] d\tau, \\
\Qb^{(0)}(t,S)&= B_{10}(t,T) h(S) + \int^T_t \hnu_{10}(\tau) B_{10}(t,\tau) \Pb^{(0)}(\tau,S) \,d\tau,
\end{aligned}\right. \\ \label{eq: bp1}
&\left\{ \begin{aligned}
\Pb^{(1)}(t,S)&=- \int^T_t \hnu_{01}(\tau) B_{01}(t,\tau) \tE_{t,S} \left[\frac{1}{2}\left(\Qb^{(0)}(\tau,S_\tau)-\Pb^{(0)}(\tau,S_\tau)\right)^2-\Qb^{(1)}(\tau,S_\tau) \right]d\tau, \\
\Qb^{(1)}(t,S)&=- \int^T_t \hnu_{10}(\tau) B_{10}(t,\tau)  \left(\frac{1}{2}\left\{\Pb^{(0)}(\tau,S)-\Qb^{(0)}(\tau,S)\right\}^2-\Pb^{(1)}(\tau,S) \right)d\tau,
\end{aligned}\right.
\end{align}
where $\hnu_{01}(t)$ and $\hnu_{10}(t)$ are specified in Proposition \ref{the_MEMM}, see \eqref{eq:Ae}.
\end{theorem}

\begin{proof}
Formally matching terms in powers of $\gamma$ in the PDE \eqref{eq: bpricepde} yields
the following equations:
\begin{subequations}\label{eq: expansion_price_pde}
\begin{align}%
&p^{(0)}_t+\frac{1}{2}\sigma_0^2 S^2p^{(0)}_{SS}+\hnu_{01}(t) (q^{(0)}-p^{(0)})=0, \\
&q^{(0)}_t+\hnu_{10}(t) (p^{(0)}-q^{(0)})=0, \\
&p^{(1)}_t+\frac{1}{2}\sigma_0^2 S^2p^{(1)}_{SS}+\hnu_{01}(t) (q^{(1)}-p^{(1)})-\frac{1}{2}\hnu_{01}(t) (q^{(0)}-p^{(0)})^2=0, \\
&q^{(1)}_t+\hnu_{10}(t) (p^{(1)}-q^{(1)})-\frac{1}{2}\hnu_{10}(t) (q^{(0)}-p^{(0)})^2=0,
\end{align}
\end{subequations}
with boundary conditions:
\begin{align*}
p^{(0)}(T,S)&=q^{(0)}(T,S) = h(S), \qquad\text{and}\quad
p^{(1)}(T,S)=q^{(1)}(T,S)=0.
\end{align*}

The equations for $p^{(0)}$ and $p^{(1)}$ are typical Cauchy problems described in \cite[Ch. 5.7]{karatzas1991brownian}. Since we have constant diffusion and volatility terms, and the source and the terminal condition functions are bounded, there exists a unique solution to \eqref{eq: expansion_price_pde} by a direct application of the Feynman-Kac representation in Theorem 7.6 in \cite{karatzas1991brownian}, leading to \eqref{eq: bp0}-\eqref{eq: bp1}. Finally, the equations for $q^{(0)}$ and $q^{(1)}$ are linear first-order ODEs that have an explicit solution using an integrating factor.
\end{proof}

Using the dual representation \eqref{eq:gamma-zero-limit}, it is immediate that $\Pb^{(0)} \equiv p^E$ and $\Qb^{(0)}$ are the MEMM prices, see Corollary~\ref{MEMMprice} and \cite{becherer2001rational}. 

\subsection{Small Probability of Liquidity Shocks}\label{sec: 3state}
Under normal conditions, the likelihood of a liquidity shock is small (typically caused by a natural or geopolitical catastrophe that is rare by definition). In terms of our model, the transition probability $\nu_{01}$ is small under this scenario and we may view the resulting option price as a \emph{perturbation} of the classical Black-Scholes price  corresponding to $\nu_{01} = 0$. In particular, for a given option maturity $T$, the probability of more than one liquidity shock is exceedingly small, $\OO(\nu_{01}^2)$.

In light of these considerations, we analyze the asymptotic case of at most a single liquidity shock. Such a restricted model may be obtained by redefining the state-space of $(M_t)$ as $\check{E}=\{0,1,2\}$. As before, we assume the market starts with the regular liquid state (0). Due to the unpredictable illiquidity shock, the market is susceptible to transiting to the illiquid state (1), and before the terminal time $T$ it may come back to the absorbing liquid state (2) which is now immunized from any further illiquidity shocks. The corresponding infinitesimal generator is
\begin{align} \label{eq:A3}
\check{A} = \begin{pmatrix}
-\nu_{01} & \nu_{ 01} & 0 \\
0 & -\nu_{10}  & \nu_{10} \\
0 & 0 & 0 \end{pmatrix}.
\end{align}

Similarly, we define the regime-dependent discounting factors to the utility functions in Merton's problem, denoted by $\check{F}_0$, $\check{F}_1$ and $\check{F}_2$.
\begin{subequations}\label{check-F}\begin{align}
\check{F}_2(t) &= F_2(t)=e^{-d_0(T-t)}, \\
\check{F}_1(t) &=\check{F}_2(t)\left ( \frac{d_0}{d_0-\nu_{10}}e^{(d_0-\nu_{10})(T-t)}-\frac{\nu_{10}}{d_0-\nu_{10}}\right ), \\
\check{F}_0(t) &= \check{F}_1(t)+\frac{d_0}{d_0+\nu_{01}-\nu_{10}}\left (e^{-(d_0+\nu_{01})(T-t)}-e^{-\nu_{10}(T-t)}\right ),
\end{align}
\end{subequations}
which are  the solutions of the ODE system
\begin{align*}
\check{F}'(t)=(\check{D}-\check{A})\check{F}(t), \quad\check{F}(T)= \begin{pmatrix}
1 \\ 1 \\ 1 \end{pmatrix}, \qquad\text{with}\quad \check{D}=\begin{pmatrix}
\frac{\mu_0^2}{2\sigma_0^2} & 0 & 0 \\
0  & 0 & 0 \\
0 & 0 & \frac{\mu_0^2}{2\sigma_0^2}\end{pmatrix}.
\end{align*}

Note that $\check{F}_2(t)=F_2(t)$, the discount factor in the classical Merton's problem. This is not a surprise considering the state (2) in this three-state liquidity switching case is in fact the classical complete market. In addition, we have $\check{F}_0(t)<\check{F}_1(t)$, and $\check{F}_2(t)<\check{F}_1(t)$. We obtain that the MEMM in the single-shock market above is given by  $\check{Q}^E={\Q}^{\check{\alpha}}$, where the minimal entropy risk factor is $\check{\alpha}(i,j)=\frac{\check{F}_j(t)}{\check{F}_i(t)}$,  $i,j \in \{0,1,2\}$. Thus, under $\check{Q}^E$ the Markov chain $(M_t)$
has transition intensities
\begin{align} 
\cnu_{01}(t)&=\nu_{01}\frac{\check{F}_1(t)}{\check{F}_0(t)}, \qquad\text{and}\qquad \cnu_{10}(t)=\nu_{10}\frac{\check{F}_2(t)}{\check{F}_1(t)}.  \label{eq: newnu10}
\end{align}

Let us denote $\cB_{ij}(s,t) \triangleq e^{-\int_s^t \cnu_{ij}(u) \,du}$, $i=0,1$ and $j=1-i$.

\begin{corollary}Using the pricing formula in Proposition \ref{prop_EMMprice}, the MEMM price $\check{p}^E$ of a contingent claim with payoff $h(S_T)$ in the single-shock market is
\begin{align}
\check{p}^E(t,S) &=\tE_{t,S} \Bigl[ \cB_{01}(t,T) h(S_T)+ \int_t^T \! \cnu_{01}(\tau)  \cB_{01}(t,\tau) \check{L}(\tau) \,d\tau \Bigr]\label{eq: MEMMprice} \qquad
\text{where } \\
\check{L}(\tau) &= \int_\tau^T{\cnu}_{10}(s)  \cB_{10}(\tau,s) h(S_{T-s+\tau}) \, ds + \cB_{10}(\tau,T) h(S_\tau).
\end{align}
\end{corollary}

Next, the corresponding indifference buyer's price assuming initial regime $M_t=0$ is the unique classical solution of the semi-linear PDE
\begin{align}\label{eq:check-p}
\check{p}_t + \frac{1}{2}\sigma_0^2 S^2\check{p}_{SS}-\frac{\nu_{01}}{\gamma}\frac{\check{g}(t,S)}{\check{F}_0}e^{\gamma \check{p}}+\frac{(d_0+\nu_{01})}{\gamma}-\frac{1}{\gamma}\frac{(\check{F}_0)_t}{\check{F}_0}=0,
\end{align}
with terminal condition $\check{p}(T,S)=h(S)$ and the explicit source term
\begin{align*}
\check{g}(t,S) =  \int_0^{T-t}\nu_{10}e^{-\nu_{10} u}e^{-d_0(T-t-u)}e^{-\gamma P_{BS}(T-t-u,S)} du+e^{-\nu_{10}(T-t)}e^{-\gamma h(S)}.
\end{align*}

As in last subsection, we can further expand in $\gamma$ to obtain
$\check{p}(t,S)=\check{p}^{(0)}(t,S)+\gamma \check{p}^{(1)}(t,S) + \ldots$, with
\begin{align} \label{eq: three_bp0p1}
\check{p}^{(0)}&(t,S)= \cB_{01}(t,T)  P_{BS}(T-t,S) +
\int^T_t \cnu_{01}(\tau) \cB_{01}(t,\tau)  \nonumber \\ & \times \left \{\int^T_\tau \cnu_{10}(\zeta) \cB_{10}(\tau,\zeta) P_{BS}(T-\zeta,S)  d\zeta + \cB_{10}(\tau,T) P_{BS}(\tau-t,S) \right \}d\tau\\
\check{p}^{(1)}&(t,S)=-\frac{1}{2}\mathbb \tE_{t,S} \Big[ \int^T_t \cnu_{01}(\tau) \cB_{01}(t,\tau) \times \Bigl\{\int^T_\tau \cnu_{10}(\zeta) \cB_{10}(t,\zeta) \bigl(\check{p}^{(0)}(\tau,S_\tau)\nonumber \\ & \quad -\tE_{\zeta,S_\zeta}[h(S_T)] \bigr)^2\, d\zeta+ \cB_{10}(\tau,T) (\check{p}^{(0)}(\tau, S_\tau)-h(S_\tau))^2\Bigr\}d\tau \Big].
\end{align}

Compared to the original model we not only have a single semi-linear PDE for the indifference price $\check{p}$ but also explicit expansion terms
 $\check{p}^{(0)}$ and  $\check{p}^{(1)}$. The zero-order term $\check{p}^{(0)}$ can be expressed as the weighted average of  the Black-Scholes prices of options with same initial condition but varying maturities,
 \begin{align}\label{eq: check-w}
 \check{p}^{(0)}(t,S) = \int_t^T P_{BS}(T-u,S) \check{w}(du)
 \end{align}
 for a probability measure $\check{w}(\cdot)$ on $[t,T]$ defined implicitly in \eqref{eq: three_bp0p1} via the transition intensities $\cnu_{ij}(t)$. Similarly,
$\check{p}^{(1)}$ is clearly negative and is a weighted average of squared differences $\check{p}^{(0)}(u,S)-P_{BS}(T-u,S)$ with \emph{same} weights $\check{w}(du)$ on $[t,T]$ as in \eqref{eq: check-w}.

\section{Implied Time-to-Maturity}\label{sec:implied-ttm}

In our liquidity switching model, the random trading interruption is the only extra risk in comparison to the Black-Scholes model. As the market is static during the no-trading time periods, one could naturally regard the model as a stochastically shortened time-to-maturity Black-Scholes model. In this subsection, we accordingly define the concept of implied time-to-maturity as the measure of the investor's view of effective time horizon based on her pricing approach. Let us first define the \emph{realized} time-to-maturity $\T$,
\begin{align}
\T \triangleq \int_0^T \ind_{\{M_s = 0\}} ds.
\end{align}
Then $\T$ is a bounded $\mathcal{G}_T$-measurable random variable. Due to the independence of $(W_t)$ and $(M_t)$ and using the strong Markov property of Brownian motion, we can view the stock as following \eqref{eq:gbm}, and the contingent claim as paying $h(S_\T)$
at the random terminal date $\T$.

\begin{definition} The \textbf{Adjusted TTM} $\tilde{T}$ under a measure $\Q$ is defined as
\begin{align}
\tilde{T}^i(T; \Q) \triangleq \E^\Q[ \T \vert M_0=i], \qquad i=0,1,
\end{align}
i.e., the original maturity $T$ subtracted by the expected no-trading time duration.
\end{definition}

Assuming a time-stationary generator $A$ as in \eqref{eq: A}, standard analysis of the sojourn times of the 2-state Markov chain $(M_t)$ yields
\begin{corollary} \label{cor_adjTTM}The adjusted time-to-maturity under the minimal EMM $\Qm$ is given explicitly by
\begin{align}\label{eq:tildeT} \left\{ \begin{aligned}
\tilde{T}^0(T; \Qm) &=\frac{1}{(\nu_{01}+\nu_{10})^2} \left[ \nu_{01} +\nu_{10}(\nu_{01}+\nu_{10})T-\nu_{01} e^{-(\nu_{01}+\nu_{10})T} \right],\\
\tilde{T}^1(T; \Qm) &=\frac{1}{(\nu_{01}+\nu_{10})^2} \left[ -\nu_{10}+\nu_{10}(\nu_{01}+\nu_{10})T+\nu_{10} e^{-(\nu_{01}+\nu_{10})T} \right].
\end{aligned}
 \right. \end{align}
\end{corollary}

Simple algebra arguments lead to $0<\tilde{T}^1(T)<\tilde{T}^0(T)<T$.
The adjusted TTM simply subtracts the expected period of market freeze, without regard for \emph{when} it will occur. Thus, it does not account for option-specific features whereby the \emph{timing of the shock} is also important, as the option time-decay depends on $t$. This motivates the following definition (compare to the popular concept of implied volatility).

\begin{definition}
Given model-based price $p=p(t,S)$, the \textbf{implied time-to-maturity} $\hat{T}(t,S; p)$ is defined implicitly via
\begin{align}\label{def: hatT}
\Pbs(\hat{T}(t,S),S) \triangleq p(t,S).
\end{align}
\end{definition}
In other words, $\hat{T}$  is defined as the time-to-maturity that equalizes $p$ with the classical Black-Scholes price.

\begin{remark} Certainly this implied TTM concept is meaningful only when \eqref{def: hatT} has a unique solution, i.e.~when the time value of an option is meaningful, such as for Calls and Puts. In such cases, we could use $\hat{T}$ to parameterize model option prices. However, for options, such as digital Calls, whose Black-Scholes price is not monotone in $T$, there is no unique implied TTM and thus this concept should not be used.
\end{remark}

The implied TTM can be used to explain the spread between the Black-Scholes price and the liquidity-adjusted valuation. Recall that the option value may be decomposed into intrinsic value (payment an agent would gain if the option is  exercised immediately) and time value (the value of  delaying exercise until maturity). Typically, the time value is increasing in time-to-maturity. Consequently, since a potential liquidity shock shrinks the expected TTM, the option price decreases (note that the intrinsic value is unaffected due to our assumption that $(S_t)$ is static during the shock) due to the extra time-decay.

\section{Numerical Examples}\label{sec: numeric}
In this section we present a series of numerical experiments to illustrate our derivations in the liquidity switching model. Recall that we have obtained three main types of prices, the indifference price $\Pb$ from \eqref{eq: bpricepde}, the MEMM price $p^E$ in Theorem \ref{MEMMprice} (equivalent to $p^{(0)}$ in Theorem \ref{p0p1}), and the minimal martingale measure price $p^{MM}$. The indifference price is obtained through utility maximization arguments, while the other two are from special martingale measures. Both $p^E$ and $p^{MM}$ are easily computed from the generic formula \eqref{eq: allEMM} using the explicit descriptions of the corresponding transition intensities of $(M_t)$. We compute $\Pb$ by solving the corresponding semi-linear PDE using a mixed implicit-explicit finite difference scheme, see \ref{app_scheme} for details.

Table \ref{tbl: params} summarizes the parameter values used. Note that with $\nu_{01}=1$ and $\nu_{10}=12$ we assume that liquidity shocks occur at a rate of once per year and last an average of one month. As two representative types of European options we  consider vanilla Calls with payoff $h(S_T) = (S_T-K)_+$ and digital Calls with payoff $h(S_T) = \ind_{\{ S_T > K \}}$.

\begin{table} [h]
\caption{\label{tbl: params} Parameter values used}
\begin{tabular}{@{}llr@{}} \hline
Parameter & Meaning & Value \\ \hline
$\mu_0$ & \text{Growth rate of the stock} & 0.06 \\
$\sigma_0$ & \text{Volatility of the stock} & 0.3 \\
$\nu_{01}$ & \text{Transition rate from state} $0 \rightarrow 1$ & 1 \\
$\nu_{10}$ & \text{Transition rate from state} $1 \rightarrow 0$ & 12 \\
$K$ & \text{Strike Price} & 10 \\
$T$ & \text{Option's maturity in years} & 1 \\
$\gamma$ & \text{Risk aversion parameter} &  1\\
 \hline
\end{tabular}
\end{table}

\subsection{Price Comparison}
We first discuss the observed relations between the different model prices $p$, $p^E$, and $p^{MM}$.
As the MMM price $p^{MM}$ assigns zero illiquidity risk premium, we have $p^E<p^{MM}$, see Proposition~\ref{the_MEMM}. Also due to risk aversion, $\Pb < p^E$, so that we have
$$p < p^E < p^{MM}.$$
Numerically, we find that at the given frequency of liquidity shocks, the difference between the MMM and MEMM prices, $p^{MM}$ vs.~$p^E$, is negligible across all option types we tried, showing that the entropic liquidity risk premium is generally very small.
 Next, our numerical results confirm that $p^E >  p$. However, the magnitude of the price spread varies for different types of options. For Call options, there is no significant difference (less than $0.5\%$ in all cases we tried);  while significant spread is found for digital options, see Figure \ref{fig: MEMM_indiff_vs_s0_for_t}. In absolute terms, the spreads are largest at-the-money (right panel of Figure  \ref{fig: MEMM_indiff_vs_s0_for_t}), however since the Call price is increasing with respect to $S$, the percent difference is generally decreasing in $S$ (left panel of Figure  \ref{fig: MEMM_indiff_vs_s0_for_t}).

Regarding the price spreads versus time-to-maturity (TTM), intuitively, as TTM shrinks, there is less chance to have liquidity shocks, thus smaller spreads should be expected. This pattern is observed for Call options. However, for digital Calls, the effect of TTM on $p^E-p$ is ambiguous. This is due to the digital Call's large Theta ($\frac{\partial p}{\partial T}$) around the strike. Hence, for short-maturity ATM digital Calls a large liquidity cost is observed, see the hump at $t=6/12$ in the right panel of Figure \ref{fig: MEMM_indiff_vs_s0_for_t}. Moreover, out-of-the-money, the digital Call B-S price is not monotone in time-to-maturity. Therefore, the shorter TTM ($t=6/12$), the smaller $p^E$, and thus the highest percent spread happens at the deep OTM case, see left panel of Figure \ref{fig: MEMM_indiff_vs_s0_for_t}.

\begin{remark}\label{rem:worst-case}
In this paper we have chosen a probabilistic description of liquidity shocks through the Markov chain $(M_t)$. An alternative formulation could focus on worst-case analysis, treating the shock interval $[\tau_{01},\tau_{10}]$ as unknown constants ,see \cite{Menkens2006crash}. In such min-max setting, the liquidity spread could be large even as $t \to T$ if the option time-decay is strong.
\end{remark}

 \begin{figure} [ht]
\centering \hspace*{-10pt}
\begin{tabular}{p{2.5in}p{2.5in}}
\includegraphics[width=2.6in,height=2in]{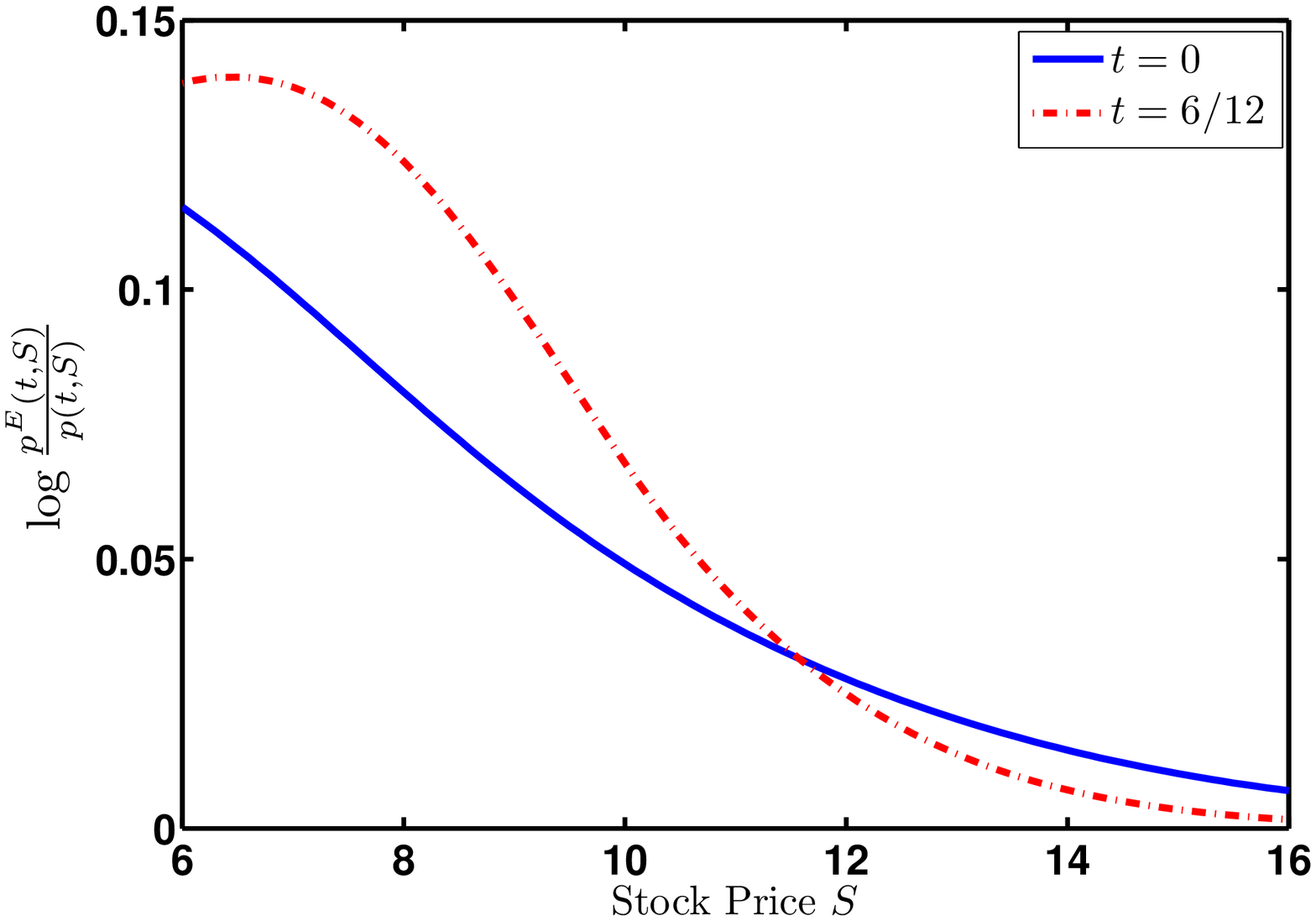}
&
\includegraphics[width=2.6in,height=2in]{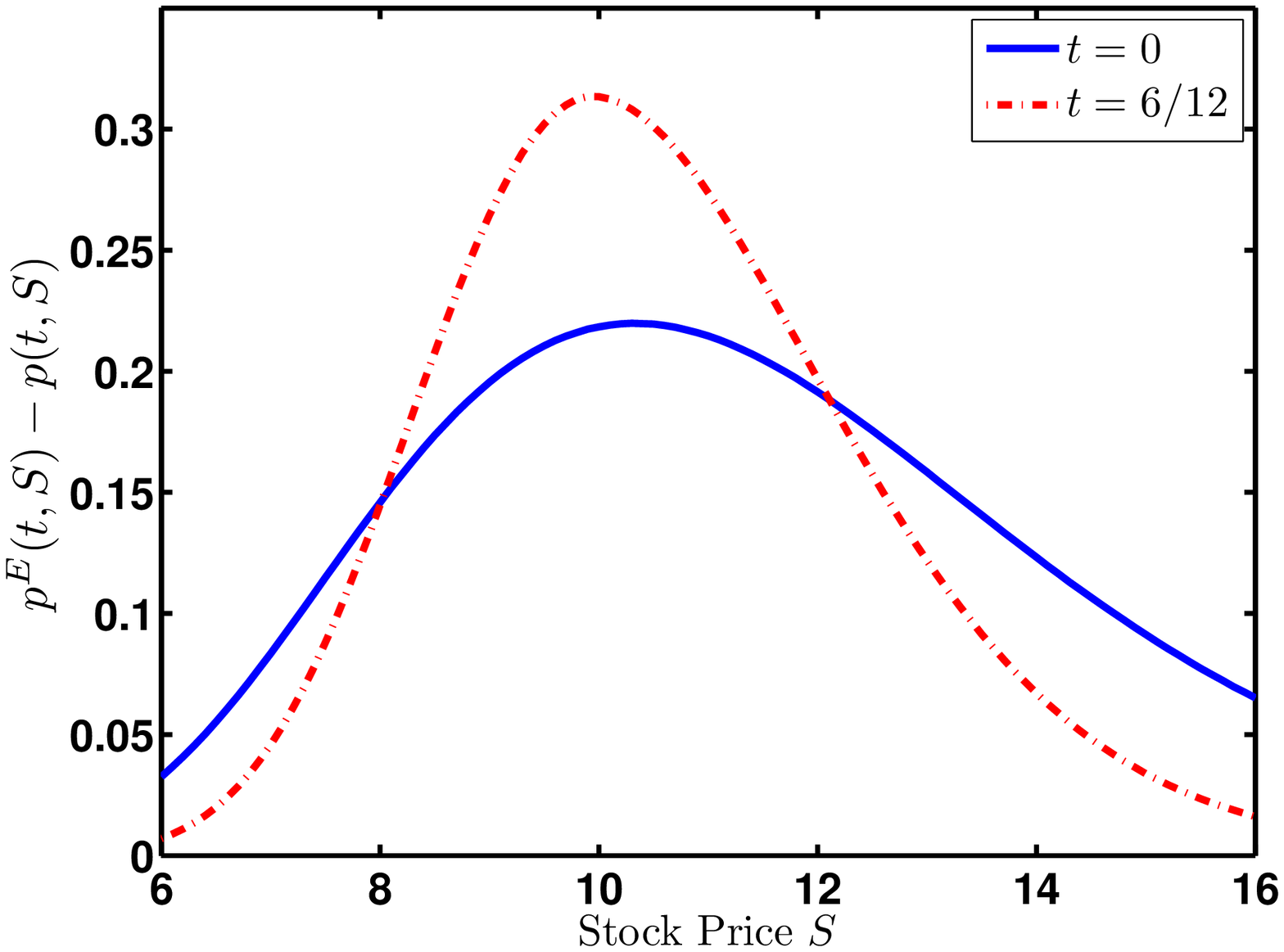}
\\
\end{tabular}
\centering
\caption{Price difference between MEMM $p^E(t,S)$ and indifference price $p(t,S)$ in percent (left panel) and in dollars (right panel) for buying 10 digital Calls. Parameters are from Table \ref{tbl: params}.}
 \label{fig: MEMM_indiff_vs_s0_for_t}
\end{figure}

Next, we study the quality of the asymptotic expansion in $\nu_{01}$ in Section \ref{sec: 3state} by comparing the indifference option price $p(t,S)$ and its approximation $\check{p}(t,S)$ from \eqref{eq:check-p}. Recall that $\check{p}(t,S)$ corresponds to the price assuming at most one liquidity shock on $[0,T]$ and is therefore the first-order term in the asymptotics using $\nu_{01}$. We find that this first-order approximation explains over 50\% of the spread across most stock prices (see Figure \ref{fig: indiff3_indiff2_vs_s0_for_t}) and can therefore provide a useful rule-of-thumb without the need to solve any nonlinear PDEs.

\begin{figure} [h]
\centering
\begin{tabular}{cc}
\includegraphics[width=4in]{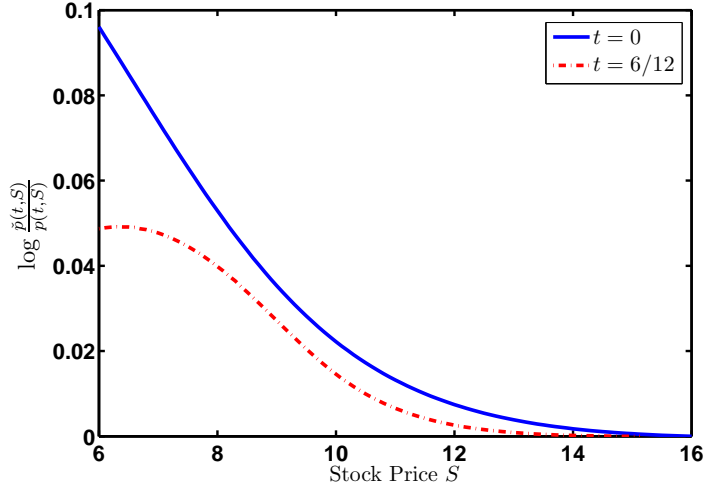}
\, \\
\end{tabular}
\centering
\caption{Small $\nu_{01}$ asymptotics. We plot the percent price difference between $\check{p}(t,S)$ defined in \eqref{eq:check-p} and $p(t,S)$ for buying 10 digital Calls. All the option and model parameters are from  Table \ref{tbl: params}.}
 \label{fig: indiff3_indiff2_vs_s0_for_t}
\end{figure}

\begin{table} [ht]
\caption{\label{tbl: price} Prices per option contract. Positive number of contracts $n$ refers to buyer's prices and negative $n$ to writer's prices. The parameters used are from Table \ref{tbl: params}. }
{\begin{tabular}{@{}c|ccc|ccc@{}}\hline 
\multicolumn{7}{c}{Indifference Price per Contract}\\
\hline \hline
\multicolumn{1}{c|}{\# of Contracts $n$} &  \multicolumn{3}{|c|}{\bf Call Option} & \multicolumn{3}{|c}{\bf Digital Call Option}\\
\multicolumn{1}{c|}{in Possession} & $S_0=8$ & {$S_0=10$} & \multicolumn{1}{c|}{$S_0=12 $}  & $S_0=8$ & $S_0=10$ &  \multicolumn{1}{c}{$S_0=12 $}\\ \hline 
10 	&	0.2875	& 	1.0720	&	2.4476	&	0.1655	&	0.4229	&	0.6705 \\
5	&	0.3128	& 	1.1264	&	2.4872	&	0.1751	&	0.4370	&	0.6826 \\
1  	& 	0.3222	&	1.1442	&	2.5014	&	0.1793	&	0.4433	&	0.6883 \\
-1 	& 	0.3253	&	1.1496	&	2.5060	&	0.1811	&	0.4461	&	0.6909 \\
-5 	& 	0.3296	&	1.1573	&	2.5124	&	0.1856	&	0.4535	&	0.6980 \\
-10	&	0.3333	&	1.1635	&	2.5178	&	0.1967	&	0.4723	&	0.7155 \\
\hline  
\multicolumn{7}{c}{Risk Neutral Price}\\
\hline \hline
MEMM	$p^E(0,S)$ 	&	0.3235 	&	1.1466 	&	2.5034 	&	0.1801 	&	0.4447 	&	0.6897\\
MMM $p^{MM}(0,S)$		&	0.3236 	&	1.1467 	&	2.5035 	&	0.1801 	&	0.4447 	&	0.6897\\
Adj B-S $P_{BS} (\tilde{T^0}(T),S)$	 &	0.3247 	&	1.1496 	&	2.5053 	&	0.1798 	&	0.4425 	&	0.6865\\
B-S  $P_{BS}(T,S)$	 	&	0.3534 	&	1.1924 	&	2.5441 	&	0.1857 	&	0.4404 	&	0.6764\\
\hline \hline
\end{tabular}}
\end{table}

To understand the nonlinear nature of the indifference pricing rule, Table \ref{tbl: price} summarizes the per-contract prices as we vary the quantity of contracts to purchase/sell. As expected, the indifference price decreases as the number of contract grows due to the increasing risk of holding more options. Also, the linear risk neutral prices (MEMM, MMM and adjusted B-S) are between the buyer and writer's prices for $n=\pm1$ (recall that MEMM is the limiting per-contract price as $n \to 0$). Recall that with exponential utility contract volume and risk-aversion asymptotics are equivalent, so Table \ref{tbl: price} can also be interpreted in terms of modifying the risk-aversion of the investor.

\subsection{Implied Time To Maturity}

As we discussed in Section \ref{sec:implied-ttm}, it is tempting to characterize the model prices using the Black-Scholes equivalent (i.e.~implied) time-to-maturity defined in (\ref{def: hatT}). As the Black-Scholes digital option price is not a monotone function of  TTM, it will only make sense to discuss Call/Put options' implied TTM.  The left panel of  Figure \ref{fig: call_ITTM} shows the adjusted $\tilde{T}^0(T-t; \Qm)$ and implied TTM $\hat{T}(t,S)$ as we vary the overall time-to-maturity $T-t$. We observe that $\tilde{T}^0(T-t)$ is almost linear in $t$ and provides a reasonable approximation to the indifference TTM $\hat{T}(t,\cdot)$. This implies that using an adjusted TTM within the classical Black-Scholes formula can be used as a simple correction to take into account liquidity risk.

The adjusted TTM $\tilde{T}$ is generally larger than the implied $\hat{T}$ due to the risk aversion and non-constant option time-decay. However, we find that this relationship is not universal; in fact for deep ITM or OTM options it may be reversed. This is illustrated in the right panel of Figure~\ref{fig: call_ITTM} that shows dependence of the implied TTM on the initial stock price $S_0$.
We observe a ``smile'' shape such that the most time-value lost is for at-the-money options. This is consistent with intuition that liquidity risk is related to the Charm and Gamma greeks which are largest at-the-money.

\begin{figure} [ht]
\centering \hspace*{-10pt}
\begin{tabular}{p{2.5in}p{2.5in}}
\includegraphics[width=2.6in,height=2in]{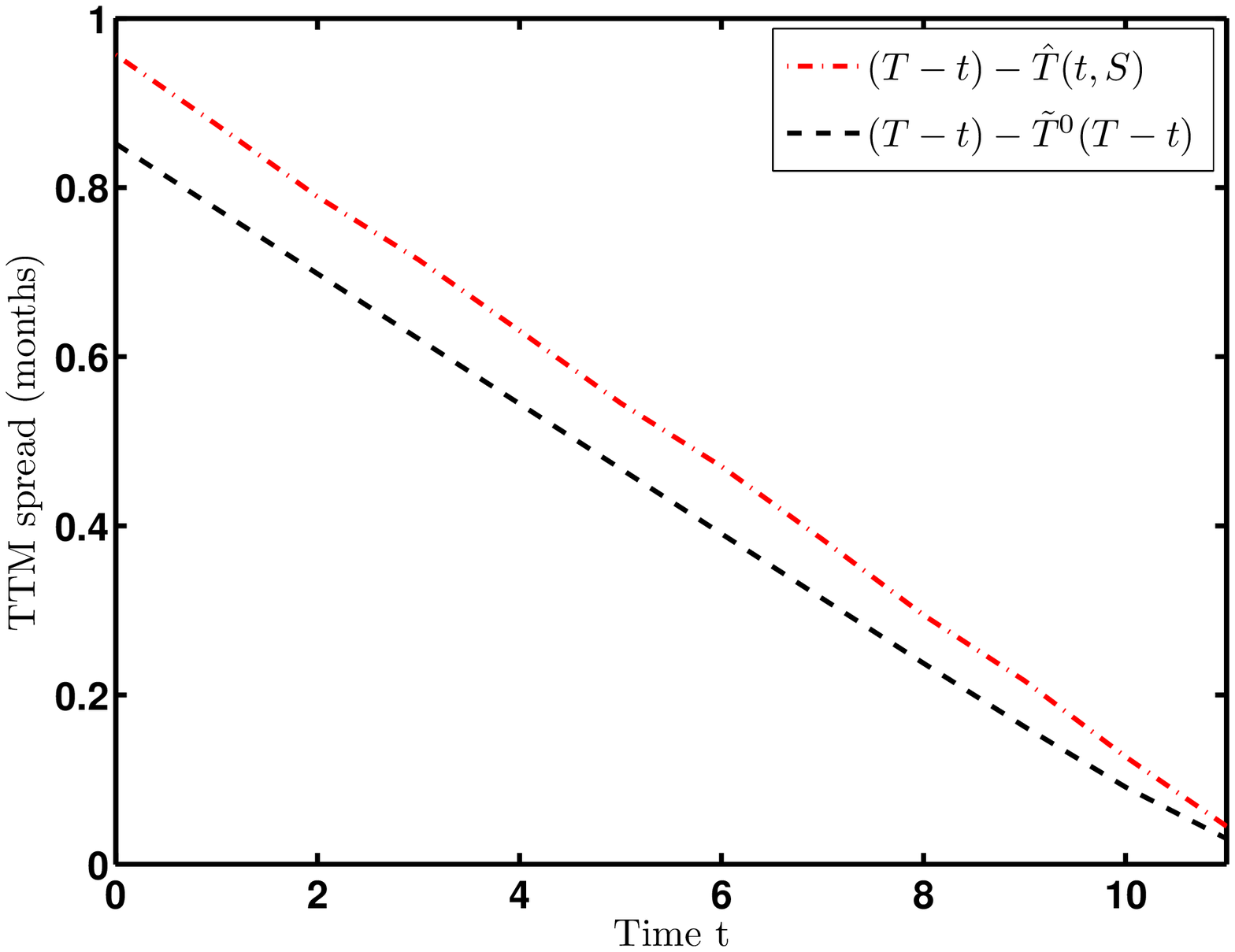}
&
\includegraphics[width=2.6in,height=2in]{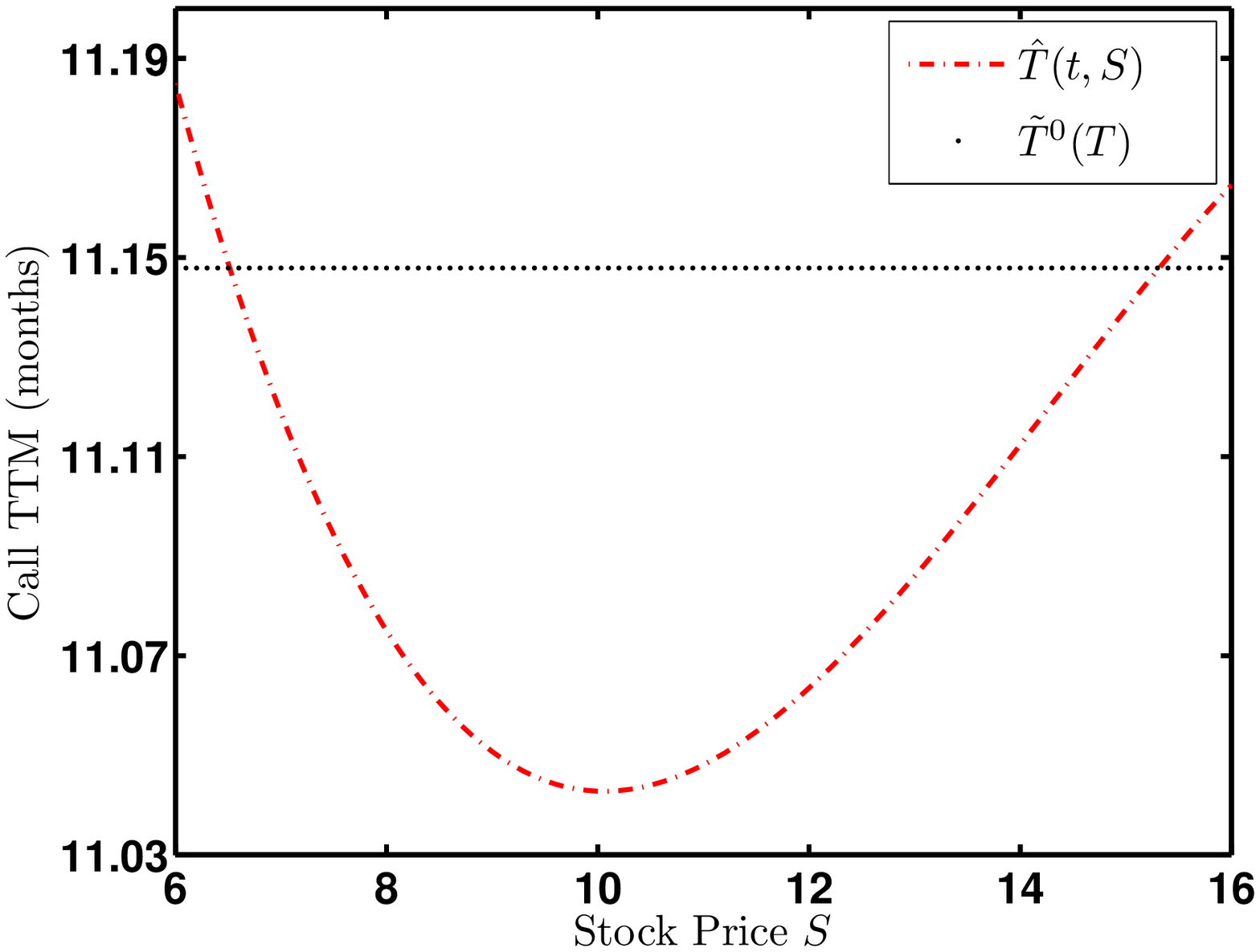}
\\
\end{tabular}
\centering
\caption{Adjusted/implied TTMs for a European Call. The parameters are as shown in Table \ref{tbl: params}. In the left panel we take $S_0 = 10$; in the right-panel we have $t=0$.}
\label{fig: call_ITTM}
\end{figure}

\subsection{Hedging}\label{sec:hedge}
 Recall that in the problem defining $U^i(t,X,S)$, the optimal stock holding in dollars is $\pi^* = \frac{\mu_0}{\sigma_0\gamma}-S\frac{\partial \Pb(t,S)}{ \partial S}$. The first term is the constant arising in the Merton's hedging problem, and the second term is the Delta hedge analogous to the classical Black-Scholes case except the price is the indifference price $p(t,S)$. Figure \ref{fig: recur_dig_delta_s0} compares this indifference hedge to the classical B-S Delta ${\Delta}_{BS}(T,S)\triangleq \frac{\partial P_{BS}(T,S)}{\partial S}$ using the original time-to-maturity $T$,  or the adjusted time-to-maturity $\tilde{T}^0(T-t; \Qm)$ from \eqref{eq:tildeT}.
\begin{figure} [h]
\centering
\includegraphics[width=4in]{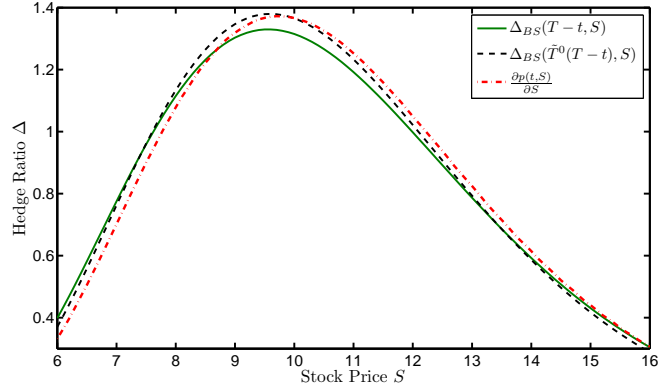}
\centering
\caption{The Deltas $\Delta_{BS}(T-t,S)$, $\Delta_{BS}({\tilde{T}}^0(T-t),S)$ (adjusted B-S delta) and $\frac{\partial p(t,S)}{\partial S}$ (indifference price delta) at $t=0$ for buying 10 digital Calls with strike $K=10$. We compute the sensitivity $\frac{\partial p}{\partial S}$ directly from the PDE \eqref{eq: bpricepde} using a first-order finite-difference approximation.
 The parameters are as shown in Table \ref{tbl: params}.}
\label{fig: recur_dig_delta_s0}
\end{figure}

Using $\Delta_{BS}$ to denote the classical Delta in the Black-Scholes model, and recalling that $p(t,S) =: P_{BS}(\hat{T}(t,S),S)$ we decompose
\begin{align*}
\frac{\partial p(t,S)}{\partial S} 
& = \Delta_{BS}(T-t,S) + \underbrace{
\Delta_{BS}(\tilde{T}^0(T-t),S) - \Delta_{BS}(T-t,S) 
}_{\text{adjusted TTM spread}} \\ & \quad +  \underbrace{
\Delta_{BS}(\hat{T}(t,S),S)-\Delta_{BS}(\tilde{T}^0(T-t),S)
}_{\text{implied TTM spread} }  + \underbrace{
\frac{\partial p(t,S)}{\partial S}- \Delta_{BS}(\hat{T}(t,S),S)
}_{\text{TTM smile correction}}.
\end{align*}
The three terms in the above comparison between the classical Black-Scholes Delta and the indifference hedge correspond to adjusting for expected time-value loss, further correcting for timing of the shocks and the risk-aversion (by using implied TTM $\hat{T}$) and finally a remainder term due to the fact that $\hat{T}(t,S)$ itself depends on $S$. Using $\Delta_{BS}(T_1,S) - \Delta_{BS}(T_2,S) \simeq \frac{\partial \Delta_{BS}(T_1,S)}{\partial T}(T_1 - T_2)$, we may then relate the Delta spread to the joint effect of the \emph{charm}  (or Delta decay, measuring the instantaneous rate of change of delta over the passage of time) and the TTM spread. Since we have that $\tilde{T}^0(T-t) < T -t$ and the digital Call charm is positive unless the option is deep ITM/deep OTM, we have
\[
\Delta_{BS}(\tilde{T},S) > \Delta_{BS}(T,S), \qquad S/K \in [0.7,1.3].
\]
The sign of $\frac{\partial p(t,S)}{\partial S} -\Delta_{BS}(\tilde{T},S)$ is more complicated, since implied TTM is ill-defined for digital options due to the changing sign of Theta. Overall, we observe that liquidity shocks make the hedges more extreme; thus  for OTM options, the risk-averse investor tends to hold fewer shares of stock; for ITM options the investor tends to hold more. An intuitive understanding is that the illiquidity shock reduces variance of $S_T$, so starting in-the-money one is more likely to end ITM, and hence the investor should hold more shares there, and vice versa.

\section{Conclusion}
In this paper, we considered the valuation and hedging problem of European options in a regime-switching market subject to liquidity shocks.
Investors exposed to such illiquidity risk will generally demand a positive risk premium (i.e.~offer to pay less to purchase options).
Our numerical experiments suggest that the entropic risk premium associated with the MEMM $\Qe$ is negligible in our model. On the other hand, indifference prices may impose significant liquidity premia depending on the option type. While these are small for vanilla Calls and Puts, we observed up to 10\% price differences for digital options (namely OTM digital Call).
For options with monotone time-value a simple way to adjust for liquidity risk is available through the adjusted TTM formalism which allows continued use of the Black-Scholes formulas after modifying the option's time-to-maturity. We believe this idea is a useful guide for traders who want a transparent rule to adjust their valuations from the base-case. In terms of the resulting hedging adjustments, the illiquidity threat causes the investor to hold more shares in-the-money and fewer shares out-of-the-money due to the reduced time-to-maturity; this correction can be approximated using adjusted Delta as well, and more generally by looking at the resulting TTM impact and the option Charm.

Even after several simplifying assumptions, our model remains complex. In \ref{app:L} we revisit a more general formulation of the impact of liquidity on asset dynamics; as can be seen in such generality little can be said beyond numerical studies. Consequently, we recommend to start with a well-understood model (e.g.~Black-Scholes) and then incrementally adjust for other risks (such as risk aversion, jump risk, illiquidity risk, etc.); in that case our analysis here, isolating impact of trading interruptions, can be used as one of such steps in obtaining the final risk-adjusted option value.

\appendix
\section{Proof of Proposition \ref{prop_EMMprice}} \label{app_EMMprice}
Under a general EMM ${\Q}^{\tilde{\alpha}}$, at time $t$, we prove that the risk neutral price $p^{\tilde{\alpha}}(t,S,0)$ with initial state $S_t = S, M_t=0$ can be written in the form (\ref{eq: gEMMprice}). Similar arguments for $p^{\tilde{\alpha}}(t,S,1)$ will be omitted.


Let $\tilde{\nu}_{ij} \triangleq \alpha_t(i,j) \nu_{ij}$ be the transition rates of $(M_t)$ under $\Q^\alpha$.
Denote by $\tau_{01}$ (resp.~$\tau_{10}$) the first transitioning time from liquid state 0 to illiquid state 1 (resp.~from state 1 to state 0), which have intensities $\tnu_{01}$ and $\tnu_{10}$. Therefore, starting with $M_t=0$, we have two possible cases, represented by $\{\tau_{01}>T\}$, and $\{t<\tau_{01}<T\}$ leading to the following expansion of the price
\begin{align*}
\E^{\tilde{\alpha}}_{t,S,0} [h(S_T)]&= \E^{\tilde{\alpha}}_{t,S,0}\left[ h(S_T) \ind_{\{\tau_{01} > T\}}+ \E^{\tilde{\alpha}}_{\tau_{01},S_{\tau_{01}},1} [h(S_T)]\ind_{\{t < \tau_{01} \le T\}} \right] \nonumber \\
&= {\Q}^{\tilde{\alpha}}_{t,S,0}(\tau_{01}>T)  \E^{\tilde{\alpha}}_{t,S}\left [ h(S_T)\right] + \int^T_t g_{01}(\tau) \E^{\tilde{\alpha}}_{\tau,S_{\tau},1} [h(S_T)] \, d\tau \nonumber
\end{align*}
where the last equality is due to the independence of $(S_t)$ and $(N_{01}(t))$ and $g_{01}(\cdot)$ is the density of $\tau_{01}$. Substituting
$
{\Q}^{\tilde{\alpha}}(\tau_{01}>T)
=\exp\bigl(-\int^T_t \tnu_{01}(u)du \bigr),
$
and
$$
g_{01}(\tau)=\frac{\partial}{\partial s}{\Q}^{\tilde{\alpha}}_{t}(\tau_{01}\le s) \Big\vert_{s=\tau}=\tnu_{01}(\tau)\exp\Bigl( -\int^\tau_t \tnu_{01}(u)du\Bigr),
$$
yields \eqref{eq: gEMMprice}.

\section{Proof of Proposition \ref{prop_hatU} and Theorem \ref{thm:classical-soln} } \label{app_HJB}
Let us first sketch out the main four steps of the proof. First, through the standard stochastic control arguments, we derive the following HJB equation for  $\hat{U}^i(t,X,S)$,
\begin{align} \label{eq: Uhat}
\hat{U}^0_t &+ \mu_0S \hat{U}^0_S + \frac{1}{2}\sigma_0^2S^2 \hat{U}^0_{SS} +\nu_{01}(\hat{U}^1-\hat{U}^0) \\
 & \qquad +\sup_{\pi}\bigl\{\mu_0X\hat{U}^0_x \pi+ \frac{1}{2}\sigma_0^2X^2 \hat{U}^0_{xx}\pi^2+ \sigma_0^2SX \hat{U}^0_{Sx}\pi \bigr\}=0, \nonumber \\ \nonumber
\hat{U}^1_t&+\nu_{10}(\hat{U}^0-\hat{U}^1)=0,
\end{align}
with terminal condition $\hat{U}^i(T,X,S)=-e^{-\gamma(X+h(S))}$. Using the scaling properties of exponential utility, we consider a candidate solution of the form $\hat{U}^i(t,X,S) =-e^{\gamma X}e^{-\gamma R^i(t,S)}$.  Substituting into \eqref{eq: Uhat} yields that $R^i(t,S)$ must satisfy the equation (\ref{eq: Rhat}). The difficulty arises due to the coupled form of \eqref{eq: Rhat} which poses several technical challenges.

To circumvent this problem, in our second step, we consider an approximation scheme $R^{(i,k)}$ defined in \eqref{eq: JPDE}  for $R^{i}$ and show that the corresponding functions $R^{(i,k)}$ are bounded and are classical solutions of the respective PDEs. This is the longest part of the proof and relies on a doubly-iterative setup. Third, we provide a stochastic control representation for $R^{(i,k)}$ which shows the convergence $\lim_{k\to\infty} R^{(i,k)} \to R^i$. Finally, using the properties of $R^{(i,k)}$ we establish a verification theorem for $R^i$.

To introduce our iterative scheme, we consider a model with at most $k \in \mathbb{N}$ liquidity shocks on $[0,T]$ (a similar strategy is used in \cite{diesinger2008asset}). Formally, let
$$
K(t) := \max \{ k : \tau_{01}^k \le t \}
$$
denote the number of liquidity shocks by time $t$. Then we define a counting process $(N_{01}^k(t))$ with intensity $\nu_{01} \ind_{\{ M_t = 0\}} \ind_{\{ K(t) < k \}}$ and use the superscript $k$ to denote the resulting wealth process $(X^{k}_t)$ and stock process $(S^k_t)$. In analogue to the original $\hat{U}^i (t,X,S)$ we then consider the value functions
 $\hat{U}^{(i,k)}(t,X,S)$ defined in \eqref{def:hatU_ik}.

\subsection*{Step 1: Control Representation}
Usual arguments imply that $\hat U^{(i,k)}$ from \eqref{def:hatU_ik} satisfy (in viscosity sense)
\begin{align} \label{eq: JPDE}
\hat U^{(0,0)}&=-e^{-d_0(T-t)}e^{-\gamma X}e^{-\gamma \tE[h(S_T)]}, \nonumber\\
\hat U^{(1,k)}_t&+\nu_{10}(\hat U^{(0,k)}-\hat U^{(1,k)})=0, \qquad k=0,1,\ldots \nonumber \\
\hat U^{(0,k)}_t&+\frac{1}{2}\sigma_0^2S^2 \hat U^{(0,k)}_{SS}+\mu_0 S \hat U^{(0,k)}_{S}+\nu_{01}( \hat U^{(1,k-1)}- \hat U^{(0,k)}) \nonumber \\
+ &\sup_{\pi}\bigl\{\frac{1}{2}\sigma_0^2X^2 \hat U^{(0,k)}_{xx}\pi^2 +\mu_0X  \hat U^{(0,k)}_{S}\pi+ \mu_0XS \hat U^{(0,k)}_{Sx}\pi \bigr\}=0, \quad k=1,\ldots,
\end{align}
where the terminal condition is $\hat U^{(i,k)}(T,X,S)=-e^{-\gamma X}e^{-\gamma h(S)}$,  $i=0,1$.

Making the ansatz
\begin{align}\label{eq: ansatz}
\hat U^{(0,k)}(t,X,S)=-e^{-\gamma X}e^{-\gamma R^{(0,k)}(t,S)}\quad\text{and}\quad \hat U^{(1,k)}(t,X,S)=-e^{-\gamma X}R^{(1,k)}(t,S)
\end{align}
and substituting the candidate optimizer $\pi^{(i,0),*}(t,X,S)=-\frac{\mu_0}{\sigma_0^2 X}\frac{\hat U^{(0,k)}_S+S\hat U^{(0,k)}_{xs}}{\hat U^{(0,k)}_{xx}}$ into equation (\ref{eq: JPDE}), we are led to consider the following iterative equations for $R^{(i,k)}$'s,
\begin{subequations}  \label{eq: RPDE}
\begin{align}
&R^{(0,0)}(t,S)=P_{BS}(T-t,S)+\frac{d_0}{\gamma}(T-t),  \label{eq: RPDE_1} \\
&R^{(1,k)}-\nu_{10}R^{(1,k)}+\nu_{10}e^{-\gamma R^{(0,k)}}=0, \quad\qquad\qquad\quad R^{(1,k)}(T,S)=e^{-\gamma h(S)}, \label{eq: RPDE_3} \\
&R^{(0,k)}_t+\frac{1}{2}\sigma_0^2S^2R^{(0,k)}_{SS}-f(t,S,R^{(1,k-1)},R^{(0,k)})=0, \quad R^{(0,k)}(T,S)=h(S), \label{eq: RPDE_5}
\end{align}
\end{subequations}
where $k=0,1,\ldots$ and
\begin{align}
f(t,S,Q,R) \triangleq\frac{\nu_{01}}{\gamma}Q e^{\gamma R}-\frac{d_0+\nu_{01}}{\gamma}.
\end{align}

Once we show in Step 2 that there exists a classical solution to \eqref{eq: RPDE_5}, our ansatz \eqref{eq: ansatz} will be validated and the control representation \eqref{def:hatU_ik} immediately follows.

\subsection*{Step 2: Classical Solutions}
The present step is to prove that equations (\ref{eq: RPDE}) have bounded classical smooth solutions $R^{(i,k)}(t,S)$'s for $t \in [0,  T]$, $S \in (0,  \infty)$, $i=0,1$, and $k < \infty$.

By assumption, $h(S)$ is bounded, so there exist constants $A,B$ such that $A \le h(S) \le B$, for all $S$.
It follows that
 $A \le P_{BS}(T-t,S) = \tilde{\E}[ h(S_T) ] \le B$, for all $t$ and $S$. Being related to the classical B-S price, $R^{(0,0)}(t,S)$ is in $C^{1,2}([0,T) \times \R_+)$ and satisfies $$A \le R^{(0,0)}(t,S) \le B+\frac{d_0}{\gamma}(T-t).
   $$ Next, $R^{(1,0)}$ follows a linear ODE in (\ref{eq: RPDE_3}), thus we have explicitly $$R^{(1,0)}(t,S)=e^{-\gamma h(S)}e^{-\nu_{10}(T-t)}+\int^T_t\nu_{10}e^{-\gamma R^{(0,0)}(u,S)}e^{-\nu_{10}(u-t)}\, du.$$ It is easy to see that  $e^{-\gamma B-d_0(T-t)} \le R^{(1,0)}(t,S) \le e^{-\gamma A},  \forall (t, S)$.

We now use induction to show that there exists a series of smooth and bounded functions $R^{(0,k)}$ and $R^{(1,k)}$, $k=1,2,\ldots$, which are the solutions to equations (\ref{eq: RPDE}). To apply the inductive argument it is sufficient to show that there exists a classical solution $R^{(0,1)}$ to the semi-linear parabolic PDE  (\ref{eq: RPDE_5}) such that $A \le R^{(0,1)}(t,S) \le B+\frac{d_0}{\gamma}(T-t)$ (same argument as for $R^{(1,0)}$ will then take care of $R^{(1,1)}$).
Typical existence and uniqueness results for solutions to semi-linear parabolic equations, such as \cite{becherer2001rational}, do not apply here because they require that the non-linear term $f(t,S,Q,R)$ be uniformly H\"{o}lder-continuous on the domain. In our equation, however, the non-linear function $f(t,S,Q,R)$ is not globally H\"{o}lder continuous in variable $R$. We therefore follow Zhou  \cite{zhou2006indifference} to resolve this difficulty through a further iterative argument. The main idea is to convert \eqref{eq: RPDE_5} into a series of linear PDEs, and using the comparison principle prove that the linear PDEs solutions converge to the solution of the original semi-linear PDE. 

First, we prove that if the solution $R^{(0,1)}(t,S)$ exists, then it is bounded.
Let us define
\[
\overline{R}(t,S) \triangleq B+\frac{d_0}{\gamma}(T-t)\qquad\text{ and }\qquad \underline{R}(t,S)\triangleq A.
\] Then
\begin{align}
&\overline{R}_t+\frac{1}{2}\sigma_0^2S^2\overline{R}_{SS}=-\frac{d_0}{\gamma} \le f(t,S,R^{(1,0)},\overline{R}),\qquad \overline{R}(T,S)=B\ge h(S), \label{R_bar}\\
&\underline{R}_t+\frac{1}{2}\sigma_0^2S^2\underline{R}_{SS}=0 \ge f(t,S,R^{(1,0)},\underline{R}), \quad\qquad \underline{R}(T,S)=A \le h(S). \label{R_under_bar}
\end{align}
So by the comparison principle stated in section (i) in the Appendix of \cite{zhou2006indifference}, we have $A \le R^{(0,1)}(t,S) \le B+\frac{d_0}{\gamma}(T-t)$ if it exists. Note that $f(t,S,Q,\cdot)$ is not globally H\"{o}lder continuous, but it is locally Lipschitz continuous in $R$ when $R \in [A,  B+\frac{d_0}{\gamma}T]$. Therefore, there exist positive constants $\underline{C}$ and $\overline{C}$ such that
\begin{align} \label{C_bar}
-\underline{C}(R^1-R^2)\le f(t,S,Q,R^1)-f(t,S,Q,R^2) \le \overline{C}(R^1-R^2)
\end{align}
for $A = \underline{R}\le R^2 \le R^1 \le\overline{R} \le B + \frac{d_0}{\gamma} T$ and $S\in \mathbb{R}^+$. Note that $\underline{C}$ and $\overline{C}$ are independent of $R^1$, $R^2$, $S$ and $t$.

Next, we construct the sequence $R^{(0,1)(j)}$ via $R^{(0,1)(0)}=\overline{R}$ and for $j=1,2,\ldots$,
\begin{align} \label{eq: iterPDE}
R^{(0,1)(j)}_t & +\mathcal{L}R^{(0,1)(j)}  =\tilde{f}(t,S,R^{(0,1)(j-1)}),  \qquad\textrm{ where } \\
 \mathcal{L}R  & \triangleq \frac{1}{2}\sigma_0^2S^2 R_{SS}-\overline{C}R, \qquad\text{and}\quad \tilde{f}(t,S,R) \triangleq f(t,S,R^{(1,0)}, R)-\overline{C}R,\nonumber
\end{align}
with terminal condition $R^{(0,1)(j)}(T,S)=h(S)$. The PDE (\ref{eq: iterPDE}) is linear, so  classical existence and uniqueness results guarantee we have a series of smooth solutions $R^{(0,1){(j)}}, j=1,2,\ldots$. Next we use induction to prove that $\underline{R} \le \cdots \le R^{(0,1)(j)} \le R^{(0,1)(j-1)} \le \cdots \le R^{(0,1)(1)} \le  R^{(0,1)(0)}=\overline{R}$.

First, we prove $\underline{R} \le R^{(0,1)(1)} \le  R^{(0,1)(0)}=\overline{R}$. From (\ref{R_bar}) and (\ref{eq: iterPDE}), we have
\begin{align}
(\overline{R}-R^{(0,1)(1)})_t+\mathcal{L}(\overline{R}-R^{(0,1)(1)}) \le 0, \nonumber\\
(\overline{R}-R^{(0,1)(1)})(T,S) \ge 0. \nonumber
\end{align}

On the other hand, similarly, we have
\begin{align}
(R^{(0,1)(1)}-\underline{R})_t+\mathcal{L}(R^{(0,1)(1)}-\underline{R})&=f(t,S,\overline{R})-\overline{C}\overline{R}+ \overline{C}\underline{R} \nonumber\\
& \le f(t,S,\overline{R})-f(t,S,\underline{R}) -\overline{C}\overline{R}+ \overline{C}\underline{R} \textrm{  (by (\ref{R_under_bar}))}\nonumber\\
& \le 0 \textrm{ (by (\ref{C_bar}))}, \nonumber \\
(R^{(0,1)(1)})(T,S)-\underline{R} &\ge 0. \nonumber
\end{align}
Therefore, by the comparison principle, we have $\underline{R} \le R^{(0,1)(1)} \le  R^{(0,1)(0)}=\overline{R}$. Now we assume $\underline{R} \le R^{(0,1)(j)} \le R^{(0,1)(j-1)} \le \cdots \le \overline{R}$, and prove that $\underline{R} \le R^{(0,1)(j+1)}\le R^{(0,1)(j)}$.

By (\ref{eq: iterPDE}), we obtain
\begin{align}
&(R^{(0,1)(j)}-R^{(0,1)(j+1)})_t+\mathcal{L}(R^{(0,1)(j)}-R^{(0,1)(j+1)})\nonumber \\
   &=\tilde{f}(t,S,R^{(0,1)(j-1)})-\tilde{f}(t,S,R^{(0,1)(j)}) \nonumber \\
&=f(t,S,R^{(1,0)}, R^{(0,1)(j-1)})-\overline{C}R^{(0,1)(j-1)}-[f(t,S,R^{(1,0)}, R^{(0,1)(j)})-\overline{C}R^{(0,1)(j)}] \nonumber\\
&=[f(t,S,R^{(1,0)}, R^{(0,1)(j-1)})-f(t,S,R^{(1,0)}, R^{(0,1)(j)}) ]-\overline{C}[R^{(0,1)(j-1)}-R^{(0,1)(j)}] \nonumber\\
& \le 0 \textrm{ (by (\ref{C_bar}))}, \qquad\textrm{ and } (R^{(0,1)(j)})(T,S)-(R^{(0,1)(k+1)})(T,S) = 0. \nonumber
\end{align}
\noindent
On the other hand, we have
\begin{align}
&(R^{(0,1)(j+1)}-\underline{R})_t+\mathcal{L}(R^{(0,1)(j+1)}-\underline{R}) \nonumber \\
&=f(t,S,R^{(1,0)},R^{(0,1)(j)})-\overline{C}R^{(0,1)(j)}+ \overline{C}\underline{R} \nonumber\\
& \le [f(t,S,R^{(1,0)},R^{(0,1)(j)})-f(t,S,R^{(1,0)},\underline{R})] -\overline{C}[R^{(0,1)(j)}-\underline{R}] \quad\textrm{  (by (\ref{R_under_bar}))}\nonumber\\
& \le 0 \textrm{ (by (\ref{C_bar}))}, \qquad\textrm{ and } (R^{(0,1)(1)})(T,S)-\underline{R}(T,S)\ge 0. \nonumber
\end{align}
\noindent
So, by the comparison principle again, we have $\underline{R}\le  R^{(0,1)(j+1)} \le R^{(0,1)(j)}$, completing  the induction on $j$.

Lastly, we prove that the sequence $\{R^{(0,1)(j)}\}$ converges to the solution of \eqref{eq: RPDE_5} for $k=1$ as $j\to\infty$. Notice that $\{R^{(0,1)(j)}\}$ is a monotone decreasing sequence, bounded below by $\underline{R}$, for all $(t,S) \in [0, T] \times \mathbb{R}^+$, hence, it converges to some $r^{(0,1)}(t,S)$ pointwise. Using the fact that $R^{(0,1)(j)}$ are bounded, it follows that $\tilde{f}(t,S,R)$ is of at most exponential growth in $S$ so by Appendix A in \cite{zhou2006indifference} we obtain that $r^{(0,1)}$ is in ${C}^{1,2}( [0, T] \times \mathbb{R}^+)$ and hence is the unique classical solution of (\ref{eq: RPDE_5}) for $k=1$. This completes the inductive step in $k$.

\subsection*{Step 3: Convergence}
To establish the uniform convergence in \eqref{eq:converge}, we apply  Theorem 3.2 in \cite{diesinger2008asset} since
the exponential utility $u(x)=-e^{-\gamma x}$ is polynomially bounded at 0. It follows that for $k$ large enough, $U^{(i,k)}$ is arbitrarily close to $U^{i}$ which is intuitive since the probability of more than $k$ shocks is exponentially small in $k$. However, since we do not have control on the rate of convergence, we cannot establish such properties for the corresponding controls $\pi^{(0,k),*}$ so it is not clear whether the hedging strategies also converge.

\subsection*{Step 4: Verification Theorem}
We are finally in position to provide a verification theorem for the original PDE \eqref{eq: Rhat}. Similar arguments to those below also provide a verification theorem for the iterative \eqref{eq: JPDE} and are omitted.

\begin{lemma}\label{lem: J} Suppose that we have $J^i(t,S,X)$ and
$\psi(t,X,S)$, 
such that
\begin{enumerate}
\item $J^{i}(t,X,S),i=0,1$ is in $C^{1,2,2}([0, T) \times \mathbb{R} \times \mathbb{R}^+)$ solve the HJB equations (\ref{eq: Rhat});

\item $J^{0}_S\sigma_0 S+J^{0}_x \sigma_0\pi X \in L^2$ for all admissible control laws $(\pi_t)$;

\item $\psi(t,X,S)$ is an admissible control;

\item For each fixed $(t,X,S)$,
$$
\hspace{-0.1in}\sup_{\pi}\Bigl\{\frac{1}{2}\sigma_0^2X^2J^{0}_{xx}\pi^2 +\mu_0X  J^{0}_{S}\pi+ \mu_0XS J^{0}_{Sx}\pi \Bigr\} =
\frac{1}{2}\sigma_0^2X^2J^{0}_{xx} \psi^2 +(\mu_0X  J^{0}_{S}+\mu_0XS J^{0}_{Sx})\psi.
$$
\end{enumerate}
Then  the optimal value function of \eqref{def:hatU_ik} is $\hat{U}^{i}(t,X,S)=J^{i}(t,X,S)$.
\end{lemma}
\begin{proof}

Given an arbitrary admissible strategy $(\pi_t) \in \mathcal{A}$, we define the process $X^{(\pi)}$ as the solution to the equation $dX^{(\pi)}=\ind_{\{M_t = 0\}}[\mu_0 \pi X^{(\pi)}dt+\sigma_0  \pi X^{(\pi)}\,dW_t]$.
Applying Ito's formula to  (\ref{eq: JPDE}), we obtain
\begin{multline}\label{eq:dJ}
dJ^{0}(t,X^{(\pi)}_t,S_t)=J^{0}_t \, dt + J^{0}_S( \mu_0S_t\, dt +\sigma_0S_t\, dW_t)+\frac{1}{2}J^{0}_{SS}\sigma_0^2S^2_t\,dt \\
                                          + J^{0}_x( \mu_0\pi X_t\, dt +\sigma_0\pi X_t\,dW_t)+\frac{1}{2}J^{0}_{xx}\sigma_0^2\pi^2X^2_t\, dt+(J^{1}-J^{0})\, dN_{01}(t).
\end{multline}
Note: for notational convenience, we omitted the parameters $(t,X^{(\pi)}_t,S_t)$ in $J^0$.

Then, in the above equation we replace $dN_{01}$ with $d\tilde{N}_{01}+\nu_{01}dt$, where $\tilde{N}_{01}$ is the compensated Poisson process associated with $N_{01}$. Thanks to the integrability assumptions, the stochastic differentials in \eqref{eq:dJ} are stochastic differentials of $\mathbb{P}$-martingales. Therefore, by using the boundary conditions of the HJB equations (\ref{eq: JPDE}), integrating up to the first liquidity shock instant $\tau^1_{01}$, and taking expectations on both sides under $\mathbb{P}$ measure, we have
\begin{align} \label{eq: UJPDE}
\E^\PP_{t,X,S}  \Big[ \ind_{\{\tau^1_{01} < T\}}   J^{1}(\tau^1_{01},& X^{(\pi)}_{\tau^1_{01}},S_{\tau^1_{01}}) + \ind_{\{\tau^1_{01} \ge T\}}u(X^{(\pi)}_T+h(S_T)) \Big]\nonumber \\ =J^{0}(t,X,S)  + \E^\PP_{t,X,S}  &\Big[\int^{\tau^1_{01} \wedge T}_t \!\! \Big\{ J^{0}_t
+\frac{1}{2}\sigma_0^2S^2J^{0}_{SS}
 +\mu_0 S J^{0}_{S}+\frac{1}{2}\sigma_0^2X^2J^{0}_{xx}\pi^2 \nonumber \\  & \qquad +\mu_0X  J^{0}_{S}\pi  +\mu_0XS J^{0}_{Sx}+\nu_{01}( J^{1}- J^{0})\pi\Big\} d\tau \Big].
\end{align}
The parameters in $J$-functions in the integral are $(\tau,X^{(\pi)}_\tau,S_\tau)$ in the right hand side of equations.
Since the $J$'s satisfy the HJB equations (\ref{eq: JPDE}), the term inside the expectation on the RHS is negative a.s., and therefore
\begin{align}\label{eq: le}
\E^\PP_{t,X,S} \left[ \ind_{\{\tau^1_{01} < T\}} J^{1}(\tau^1_{01},X^{(\pi)}_{\tau^1_{01}},S_{\tau^1_{01}}) + \ind_{\{\tau^1_{01} \ge T\}} u( X^{(\pi)}_T+h(S_T)) \right]\le&J^{0}(t,X,S).
\end{align}
Using the fact that $(X_t,S_t)$ are static on $\{ M_t = 1\}$, so that
$$
J^1(t,X,S) = \E^\PP_{t,X,S} \left[ \ind_{\{ \tau^1_{10} < T\}} J^0(\tau^1_{10},X,S) + \ind_{\{ \tau^1_{10} \ge T\}} u( X + h(S)) \right],
$$
and applying the strong Markov property of $(N_{ij}(t))$ at $\tau^1_{01}$ we obtain
\begin{align*}
\E^\PP_{t,X,S} \left[ \ind_{\{\tau_{10}^1 < T\}} J^{0}(\tau_{10}^1,X^{(\pi)}_{\tau_{10}^1},S_{\tau_{10}^1}) + \ind_{\{\tau_{10}^1 \ge T\}} u( X^{(\pi)}_T+h(S_T)) \right]\le&J^{0}(t,X,S).
\end{align*}
By induction on the transition times $(\tau_{01}^k, \tau_{10}^k)$ it follows that
\begin{align}\label{eq: le-k}
\E^\PP_{t,X,S} \left[ \ind_{\{\tau_{10}^k < T\}} J^{0}(\tau_{10}^k,X^{(\pi)}_{\tau_{10}^k},S_{\tau_{10}^k}) + \ind_{\{\tau_{10}^k \ge T\}} u( X^{(\pi)}_T+h(S_T)) \right]\le&J^{0}(t,X,S).
\end{align}
Since $\tau_{10}^k \to \infty$ a.s., and using the fact that $J^0(\cdot,X^{(\pi)}_\cdot,S_\cdot)$ is bounded from below, we conclude using the dominated convergence theorem that
\begin{align}\label{eq: le-infty}
J^{0}(t,X,S) \ge  \E^{\PP}_{t,X,S} \left[ -\exp \left( -\gamma\{X^{(\pi)}_T+h(S_T)\} \right)\right].
\end{align}

On the other hand, if we choose the control law $\psi$, the term inside the expectation in \eqref{eq: UJPDE} is zero a.s., yielding
\begin{align} \label{eq: ge}
\E^\PP_{t,X,S}\left[-\exp(-\gamma \{X^{(\psi)}_T+h(S_T)\}) \right]=J^{0}(t,X,S).
\end{align}
Combining (\ref{eq: le-infty}) and (\ref{eq: ge}) we conclude $\hat{U}^{i}(t,X,S)=J^{i}(t,X,S)$, $i=0,1$.
\end{proof}

\section{Numerical Scheme} \label{app_scheme}
To solve for the indifference price $p(t,S)$, we apply implicit finite-difference scheme to the system \eqref{eq: bpricepde}.
Working with the log-price $Z_t \equiv \log S_t$, we use an equally spaced mesh $(\iota, \mathcal{Z})$, with grid points $(t_i,Z_j)$,
$t_i = i\delta_t$, $i=0,\ldots,N$, $\delta_t = T/N$ and $Z_j = Z_{min} + j\delta_z$, $j=0, \ldots, M$, with
\begin{align}
\delta_z=\sqrt{\sigma_0^2 \delta_t + \bigl(\frac{1}{2}\sigma_0^2 \bigr)^2\delta_t^2}. \label{eq:dz}
\end{align}
Letting $p^i_j = p(t_i,z_j)$ denote the option price on the grid, and similarly for $q^i_j$,
we use the standard second-order centered finite-difference approximation (see e.g.~\cite{DewynneWilmott}) for all the first and second derivatives of $p(t,S)$ and use one-sided approximations at the boundaries $Z=Z_{min}$ and $Z=Z_{max}$.
To avoid the difficulty of the nonlinear term in the $p$-PDE, the following first order approximation is used
\begin{align}
e^{\gamma \Pb(t_i,z_j)} \equiv e^{\gamma \Pb^i_j} \simeq e^{\gamma\Pb^{i+1}_{j}} + \gamma e^{\gamma\Pb^{i+1}_{j}} (\Pb^{i}_j-\Pb^{i+1}_{j}).
\end{align}
The resulting implicit tri-diagonal system for $(p^i_j)$ is solved using the Thomas algorithm, see \cite[ Ch 3.5]{strikwerda2004finite}.
In the second half-step, we apply an explicit Euler scheme for the $q$-ODE (second line of \eqref{eq: bpricepde}) using the just-computed $\Pb^i_j$,
\begin{align}
\frac{\Qb^{i+1}_j-\Qb^{i}_j}{\delta_t} -\frac{\nu_{10}}{\gamma}\frac{F_0(t_i)}{F_1(t_i)}e^{-\gamma (\Pb^i_j-\Qb^{i+1}_j)}+\frac{\nu_{10}}{\gamma}-\frac{1}{\gamma}\frac{F_1'(t_i)}{F_1(t_i)}=0,
\end{align}
which can be explicitly solved for $\Qb^i_j$.

\section{General Illiquid Dynamics} \label{app:L}
One could choose to work with \eqref{eq:illiquid-dynamics}, the general case when $\mu_1 \ne 0$, $\sigma_1 \ne 0$, and  $L \ne 0$.
In that case the redefined value functions $\hat{U}^0(t,X,S)$ and $\hat{U}^1(t,X,S,\pi)=e^{-\gamma X}R^1(t,S,\pi)$ still allow separation of wealth, however since $\pi$ is non-constant during the liquidity shock, $R^1$ solves a two-dimensional PDE in $(S,\pi)$
that in turn creates further $\pi$-dependence in the HJB equation satisfied by $R^0$.
With the extra state variable $\pi$, we no longer have a closed-form solution for the Merton problem (i.e.~when $h(S) \equiv 0$) either. 
In analogue to the indifference price definition in (\ref{eq: bpricedef}) we can show that the indifference prices satisfy the HJB equations
(cf.~\eqref{eq: bpricepde})
\begin{align} \label{eq: bpricepde_g}
&\quad \Pb_t + \frac{1}{2}\sigma_0^2 S^2(\Pb_{SS}-\gamma \Pb_S^2)+\mu_0S\Pb_S+\frac{\nu_{01}}{\gamma}-\frac{1}{\gamma}\frac{F_0'}{F_0} \\
+&\sup_{\pi}\Bigl \{\mu_0 \pi-\frac{1}{2}\sigma_0^2\pi^2\gamma-\gamma\sigma_0^2S\Pb_S\pi-   \frac{\nu_{01}}{\gamma}\frac{F_1}{F_0}e^{\gamma \Pb}e^{\gamma L\pi}e^{-\gamma \Qb(t,(1-L)S,(1-L)\pi)}\Bigr \}=0, \nonumber \\
&\quad \Qb_t +\frac{1}{2}\sigma_1^2 S^2(\Qb_{SS}-\gamma \Qb_S^2)+\frac{1}{2}\sigma_1^2 \pi^2(\Qb_{\pi\pi}-\gamma \Qb_\pi^2)+\sigma_1^2 \pi S(\Qb_{s\pi}-\gamma \Qb_S\Qb_\pi)-\gamma\sigma_1^2\pi S\Qb_S\nonumber \\ -&\gamma\sigma_1^2\pi^2\Qb_\pi +\mu_1S\Qb_S+\mu_1\pi\Qb_\pi+\mu_1\pi-\frac{1}{2}\sigma_1^2\pi^2\gamma-\frac{\nu_{10}}{\gamma}\frac{F_0}{F_1}e^{-\gamma (\Pb-\Qb)}+\frac{\nu_{10}}{\gamma}-\frac{1}{\gamma}\frac{F_1'}{F_1}=0, \nonumber
\end{align}
with terminal condition: $\Pb(T,S)=\Qb(T,S,\pi) =h(S)$.
No analytic insights into \eqref{eq: bpricepde_g} seem possible.  If we assume a static illiquid state ($\mu_1=0, \sigma_1=0$), but $L \ne 0$, then the resulting $\Qb^{(L)}=\Qb^{(L)}(t,S)$ is independent of $\pi$ (because $d\pi_t = 0$), and the price equations (\ref{eq: bpricepde_g}) can be simplified to
\begin{align} \label{eq: bpricepde_L}
\Pb^{(L)}_t& + \frac{1}{2}\sigma_0^2 S^2(\Pb^{(L)}_{SS}-\gamma (\Pb^{(L)}_S)^2)+\mu_0S\Pb^{(L)}_S+\frac{\nu_{01}}{\gamma}-\frac{1}{\gamma}\frac{(F_0^{(L)})'}{F_0^{(L)}} \nonumber \\
&+\sup_{\pi}\bigl \{\mu_0 \pi-\frac{1}{2}\sigma_0^2\pi^2\gamma-\gamma\sigma_0^2S\Pb^{(L)}_S\pi-   \frac{\nu_{01}}{\gamma}\frac{F_1^{(L)}}{F_0^{(L)}}e^{\gamma \Pb^{(L)}}e^{\gamma L\pi}e^{-\gamma \Qb^{(L)}(t,(1-L)S)}\bigr \}=0, \nonumber \\
\Qb^{(L)}_t &-\frac{\nu_{10}}{\gamma}\frac{F_0^{(L)}}{F_1^{(L)}}e^{-\gamma (\Pb^{(L)}-\Qb^{(L)})}+\frac{\nu_{10}}{\gamma}-\frac{1}{\gamma}\frac{(F_1^{(L)})'}{F_1^{(L)}}=0.
\end{align}
In \eqref{eq: bpricepde_L} the Merton terms $F_0^{(L)}=F_0^{(L)}(t)$ and $F_1^{(L)}=F_1^{(L)}(t)$ are free of $\pi$ as well, and
are given implicitly via
\begin{align}
&\frac{dF_0^{(L)}}{dt}-\nu_{01}F_0^{(L)}-\sup_{\pi}\Bigl \{\gamma \mu_0 \pi F_0^{(L)}-\frac{1}{2} \sigma_0^2\gamma^2F_0^{(L)}\pi^2-\nu_{01}e^{\gamma L\pi}F_1^{(L)}(t) \Bigr\}=0,\nonumber \\
&\frac{dF_1^{(L)}}{dt}+\nu_{10}(F_1^{(L)}-F_0^{(L)})=0,\quad\qquad F_0^{(L)}(T)=F_1^{(L)}(T)=1.
\end{align}
We hope these derivations convince the reader that this more general route is fraught with tedious complexity and hence provide intuition for the use of our main model in Section \ref{sec:model}.

\bibliographystyle{amsplain}
\bibliography{illiquid}

\end{document}